\documentclass[a4paper,english,cleveref]{lipics-v2021}
\usepackage{xspace}
\usepackage{mathtools}
\usepackage{nicefrac}
\usepackage{amsmath, amssymb, amsthm}
\usepackage[subrefformat=simple,labelformat=simple]{subcaption}

\hideLIPIcs
\nolinenumbers

\crefname{figure}{Figure}{Figures}
\crefname{theorem}{Theorem}{Theorems}
\crefname{lemma}{Lemma}{Lemmas}
\crefname{claim}{Claim}{Claims}
\crefname{observation}{Observation}{Observations}
\crefname{corollary}{Corollary}{Corollaries}
\crefname{section}{Section}{Sections}
\crefname{definition}{Definition}{Definitions}

\newcommand{\problemtitle}{\textsc{STAP}\xspace}

\title{Particle-Based Assembly Using Precise Global Control}
\titlerunning{Particle-Based Assembly Using Precise Global Control}
\author{Jakob Keller}{Department of Computer Science, TU Braunschweig, Braunschweig, Germany}{jkeller@ibr.cs.tu-bs.de}{https://orcid.org/0000-0001-9988-953X}{}
\author{Christian Rieck}{Department of Computer Science, TU Braunschweig, Braunschweig,  Germany}{rieck@ibr.cs.tu-bs.de}{https://orcid.org/0000-0003-0846-5163}{}
\author{Christian Scheffer}{Faculty of Electrical Engineering and Computer Science, Bochum University of Applied Sciences, Bochum, Germany}{christian.scheffer@hs-bochum.de}{https://orcid.org/0000-0002-3471-2706}{}
\author{Arne Schmidt}{Department of Computer Science, TU Braunschweig, Braunschweig,  Germany}{aschmidt@ibr.cs.tu-bs.de}{https://orcid.org/0000-0001-8950-3963}{}

\authorrunning{J.~Keller, C.~Rieck, C.~Scheffer, and A.~Schmidt}
\Copyright{Jakob Keller, Christian Rieck, Christian Scheffer, and Arne Schmidt}



\relatedversion{This is the full version of an extended abstract~\cite{keller2021particle} that appeared in the 17th Algorithms and Data Structures Symposium (WADS 2021).}

\supplement{The 3D models can be accessed through: \url{https://gitlab.ibr.cs.tu-bs.de/alg/particle-based-assembly-extra}}

\funding{}

\acknowledgements{We thank Linda Kleist for valuable discussions and suggestions that improved the presentation of this paper, and Matthias Konitzny for the awesome 3D images.}

\EventEditors{}
\EventNoEds{42}
\EventLongTitle{}
\EventShortTitle{}
\EventAcronym{}
\EventYear{}
\EventDate{}
\EventLocation{}
\EventLogo{}
\SeriesVolume{}
\ArticleNo{}

\begin{document}

\maketitle

\begin{abstract}
In micro- and nano-scale systems, particles can be moved by using an external force like gravity or a magnetic field. In the presence of adhesive particles that can attach to each other, the challenge is to decide whether a shape is constructible. Previous work provides a class of shapes for which constructibility can be decided efficiently when particles move maximally into the same direction induced by a global signal.

In this paper we consider the single step model, i.e., a model in which each particle moves one unit step into the given direction. We restrict the assembly process such that at each single time step actually one particle is added to and moved within the workspace.
We prove that deciding constructibility is NP-complete for three-dimensional shapes, and that a maximum constructible shape can be approximated. The same approximation algorithm applies for 2D. We further present linear-time algorithms to decide whether or not a tree-shape in 2D or 3D is constructible.
Scaling a shape yields constructibility; in particular we show that the $2$-scaled copy of every non-degenerate polyomino is constructible. In the three-dimensional setting we show that the $3$-scaled copy of every non-degenerate polycube is constructible.
\end{abstract}

\section{Introduction}
In recent years, the easier access to micro- and nano-scale systems has given rise to challenges that deal with programmable matter. 
In some of these applications, particles can be controlled by a global external force such as gravity or a magnetic field.
On actuation, every particle moves into the same direction at unit speed.
Assembly of particles into desired structures using maximal movements, i.e., every particle moves into a given direction until it hits an obstacle or another particle, has been investigated in~\cite{hierarchical, full_tilt, becker2018tilt,schmidt2018efficient}.
However, it is also reasonable to expose the particles to these forces just for a limited amount of time, such that more precise movements become possible. Reconfiguration of a set of particles~\cite{hierarchical}, gathering all particles~\cite{gathering_swarm_medicine, gathering_swarm}, or assembling patterned rectangles~\cite{caballerobuilding} are well studied problems.

In this paper, we consider the construction of tile-based structures (such as polyominoes in 2D, and polycubes in 3D) through adhesive particles, that all move one step into the same direction through activation of the global external force. We limit ourselves to the case where only one particle at a time is added to the assembling workspace at a predefined position.
Whenever two particles come close, they stick together and no longer separate. This behavior is usually defined in terms of~\emph{glues} that are assigned to the particle's sides, and \emph{temperature}. 
	Glues are associated with a \emph{binding strength}, so that two assemblies can only stick together, when the sum of binding strengths along their common boundary exceeds the temperature, e.g., see the introductory work of Winfree~\cite{dna_self}. 
	Several different glue types as well as attachment rules are defined in the literature (see~\cref{sec:lit:selfassembly}). 
	However, in this paper all particles have the very same glue on all sides, which is why we do not use this term any further and only argue that these particles stick to each other as soon as they are adjacent. To~start the building process of a desired shape, we assume that there is a seed particle that is already placed in the workspace. 
	Furthermore, we consider the seed to be ``anchored'' to its position so that it is not moved by the external force.
	This implies that only one particle actually moves at every time step.

In this model we consider the problem of deciding whether or not a given shape~$P$ is constructible. We call this the \textsc{Single Step Tilt Assembly Problem}~(STAP). 
	For all non-constructible shapes $P$ we ask for a constructible subshape $P_{\max}\subset P$ of maximum size, i.e., $P_{\max}$ is the largest shape among all constructible subshapes of~$P$. We call this optimization variant \textsc{MaxSTAP}.

For a precise model description and fundamental definitions, see~\cref{sec:prelim}.

\subsection{Contributions}
In this paper, we focus on the problem of constructibility within in the single step model in a free workspace. We prove that the problem is NP-complete in the three-dimensional setting~(\cref{thm:hardness}). The natural optimization variant allows for a $\Omega(n^{\nicefrac{-1}{d}})$-approximation, where $d$ denotes the dimension~(\cref{thm:constructability}). We provide a linear time algorithm that decides whether or not a given tree-shape is constructible~(\cref{thm:tree}). In the case of non-constructible non-degenerate shapes, we show that the 2-scaled copy of every polyomino is constructible~(\cref{thm:scaling:simple:2D}), while similar arguments show that the 3-scaled copy of every polycube is constructible~(\cref{thm:scaling:simple:3D}). The result is tight for the 2D case, and we conjecture that a 2-scaling also suffices for the 3D case.

\subsection{Related work}\label{sec:lit:selfassembly}
We divide the related work into three categories: Algorithmic results in tile self-assembly and tilt problems, and practical approaches.

\paragraph*{Tile self-assembly}
Instead of relying on universal, external control of agents, it is possible to use DNA as material.
DNA-strands can either be used to fold into the desired shape~\cite{rothemund2005design}, or to create building blocks based on \emph{Wang tiles}~\cite{wang1961proving,wang1990dominoes} that can then again self-assemble into shapes in a non-deterministic way. This was first introduced by Winfree~\cite{dna_self} as the \emph{abstract Tile Self-Assembly Model} (aTAM). The model became a standard in theoretical work on self-assembly, as well as proved suitable in practical experiments~\cite{rothemund2000program,rothemund2001theory,schulman2007synthesis,SoloveichikW07,winfree1998design}.
Instead of considering the case where all DNA tiles are binding to a seed assembly, it is also natural to consider multiple seeds that simultaneously grow subassemblies which can then again attach to each other.
This 2-handed tile self-assembly model~\cite{cannon_2013_2handed,schweller2019nearly} can also be used to let the subassemblies connect in phases, i.e., in a hierarchical manner~\cite{parallelism_self}. One generalized version of hierarchical assembly is called ``staged tile self-assembly''~\cite{chalk2018optimal, chalk2019optimal, staged_constant_glues, self-assembly_staged}. 
Instead of relying on square-shaped tiles, some generalizations are also already considered~\cite{DemaineDFPSWW14,FuPSS12}.
For more details on algorithmic self-assembly, we refer to the surveys by Doty~\cite{doty2012theory}, and Patitz~\cite{selfassembly_survey}, and Woods~\cite{woods2015intrinsic}.

\paragraph*{Tilt problems}
The process of self-assembly works by diffusion and is non-deterministic.
If a deterministic approach is desired, we can use global control to move particles into different directions.
Using global control opens a wide field of problems, that we summarize briefly.

Becker et al.~\cite{computation_swarms} show that particles can be used for computation by implementing \textsc{not}, \textsc{nand}, \textsc{nor}, \textsc{xor} and \textsc{xnor} gates within a particularly designed obstacle environment.
If the model is generalized by adding $2\times1$ particles (dominoes), it is also possible to construct \textsc{fan-out} gates~\cite{computation3_swarms, computation2_swarms}. 

Mahadev et al.~\cite{gathering_swarm} show how to gather $n$ particles within an obstacle environment in $O(n^3)$ actuation steps.
Becker et al.~\cite{gathering_swarm_medicine} improve the runtime to only depend on the geometric complexity of the workspace, rather than on the number of particles.

Closely related to this problem are occupancy, relocation and reconfiguration problems.
For a given set of particles within an obstacle environment, these problems ask for a sequence of tilts such that (i) any particle reaches a designated position, (ii) a specific particle reaches a designated position, and (iii)~every particle reaches its respective target position.
If particles are allowed to move a unit step on actuation, Caballero et al.~\cite{caballero2020hardness} show that (i) permits a linear-time algorithm, while the decision variants of (ii) and (iii) are NP-hard. More recently, Caballero et al.~\cite{caballero2020relocating} show PSPACE-completeness for (ii).
Becker et al.~\cite{reconfiguring_swarm} show that (i) is NP-hard when particles have to move maximally.
Balanza-Martinez~et~al.~\cite{hierarchical} give tighter results by proving PSPACE-completeness for all three problems in the full tilt model. Note that these problems are hard, even if we do not allow the particles to stick together.

Manzoor~et~al.~\cite{parallel_self} provide algorithms for assembling shapes, called \emph{drop shapes}, from \emph{sticky} particles under global control.
In this assembly process, only one particle at a time is added by maximal movements to a seed assembly.
They also show that the assembly process can be pipelined, i.e., the same shape is produced multiple times.
Becker et al.~\cite{becker2018tilt} prove that every drop shape permits an amortized construction time of $O(1)$, provided that sufficiently many copies are constructed.
However, every shape needs a custom-designed obstacle workspace.
Balanza-Martinez et al.~\cite{full_tilt} 
designed a single obstacle workspace in which every drop shape can be constructed, if it fits in a $w\times h$ rectangle. Additionally, Caballero et al.~\cite{caballero2021fast} describe a universal constructor for assembling shapes in the single step model.
By using not only single particles but whole subassemblies, the class of constructible shapes increases, see~\cite{hierarchical, schmidt2018efficient}.
A~crucial step is to decompose a given shape into two connected parts, that can then be pulled apart into a single direction, without causing collisions.
Agarwal et~al.~\cite{agarwal2021two} recently proved that this problem is NP-hard, even if a direction is given.

Another problem that arises from a practical point of view is error detection. 
Keldenich et al.~\cite{sort_polyominoes} show that ortho-convex shapes can be classified within an obstacle workspace, i.e., these shapes can be tested for errors during the assembly process.

\paragraph*{Practical approaches}
On the practical side, there are multiple approaches related to the assembly of micro-structures with uniform movement. Related work on controlling a swarm of particles, especially in biological or medical environments, include control by electricity~\cite{brown1981galvanotaxic}, chemical reactions~\cite{weibel2005microoxen}, light~\cite{steager2007control, weibel2005microoxen} and magnetism.
Because magnetic resonance imaging (MRI) scanners already employ magnetism in a clinical context, it is reasonable to design a magnetic control device based on MRI scanners~\cite{MRI_moving, MRI1, mri_platform, MRI2, MRI3}.
Another advantage of using MRI scanners lies in the possibility of tracking the positions of the agents as well as moving them.
If the devices provide stronger custom coils, this approach might even lead to new applications where tissue penetration is wanted~\cite{MRI_penetration}. 
There are different types of particles discussed in the literature.
On the one hand, there is research in which single cell organisms (\emph{Tetrahymena Pyriformis}) are fed with iron particles so that they can be steered by applying a magnetic field~\cite{magnetic_singlecell3,magnetic_singlecell,magnetic_singlecell2}.
On the other hand, research considers the fabrication of artificial micro-swimmers.
These are usually called \emph{Artificial Bacterial Flagella} (ABF) because they resemble their natural counterpart~\cite{dreyfus2005microscopic}.
Another very popular approach is to design ABFs that consist of a magnetic head, that can be turned by a revolving magnetic field, and a rigid helical tail~\cite{cheang2010fabrication, ghosh2009controlled, peyer2013magnetic, peyer2013bio, zhang2009artificial, zhang2010artificial}.
By quickly turning the micro-swimmer, the tail generates a thrust force by its corkscrew shape.

\section{Preliminaries and problem definition}\label{sec:prelim}
In this section we provide a detailed description of the considered model as well as basic definitions that are used throughout the paper. If necessary, further definitions related to individual sections are given in the respective sections.

\subparagraph*{Polyomino.}
Let $P \subset \mathbb{Z}^2$ be a finite set of $n$ grid points in the plane. The grid points of $\mathbb{Z}^2$ are called \emph{positions}. The embedded graph $G_P$ is the grid graph induced by $P$, in which two vertices are adjacent if they are at unit distance. If $G_P$ is connected, we obtain a \emph{polyomino} by placing a unit square, called \emph{tile}, centered on every vertex of $G_P$. We refer to the \emph{size} of a shape by the number of its tiles. At some points, we also use $\vert P\vert$ to denote the size of a shape $P$, if this is more practical. Two tiles or positions are \emph{adjacent}, if their respective vertices in $G_P$ are adjacent, i.e., they share an edge. A position~$p\in \mathbb{Z}^2$ is \emph{occupied}, if there is a tile that is placed on $p$, and \emph{free} otherwise. The neighborhood $N[\cdot]$ of a tile or position is the respective set of adjacent positions. A polyomino $P$ is \emph{simple} if the grid graph $\mathbb{Z}^2\backslash G_P$ is connected. A polyomino is \emph{degenerate} if there exist two tiles $t_1, t_2\in P$ such that $\vert N[t_1] \cap N[t_2]\vert = 2$, and each position $p\in N[t_1] \cap N[t_2]$ is free. If $G_P$ is a tree, the respective polyomino is \emph{tree-shaped}.

\subparagraph*{Workspace.} 
The \emph{workspace} is a rectangular region of $2n\times 2n$ positions, anchored at position~$(0,0)$. The size of the workspace is proportional to the size of the shape to be constructed, so that there is enough space to all sides. The workspace contains an \emph{anchored seed tile} at position~$(n,n)$. Note~that the~term ``anchored'' does not require that the seed tile has a specific position within the shape, but only that it occupies a specific position within the workspace during the whole construction process. Thus, the seed actually is not affected by the global force, i.e., it does not move, implying that all tiles connected to the seed will not move either.

\subparagraph*{Construction step.}
A tile can \emph{move} from one position $p$ to an adjacent position~$q$ as long as the neighborhood $N[p]$ is free. A \emph{construction step} is a sequence~$\sigma$ of moves such that $\sigma$ moves a tile to a position adjacent to an occupied position, i.e., adjacent to a tile of the current assembly. Every construction step starts by \emph{adding} a new tile to the workspace.
Without loss of generality, we assume that every construction step starts at position $(0,0)$. Note that we limit a construction step to be completed, before the subsequent one can start. This means that during a construction step, only the newly added tile moves within the workspace. Whenever two tiles are at adjacent positions, they stick together and no longer separate.
Analogously, we define a \emph{deconstruction step} as a sequence~$\widetilde\sigma$ that moves a tile from its position to the position~$(0,0)$.
Note~that these are ``reverse'' moves, i.e., a tile can move from a position $p$ to an adjacent position $q$ if $N[q]$ is free.
If~there is a deconstruction step for a tile $t$, we call $t$ \emph{removable}.

We have not described the actual process of adding a tile to the workspace, as we do not use this in a real-life application. However, one idea would be to place all needed particles at a sufficiently large distance from each other so that they do not get too close to each other during a construction step. With this, one could make sure that a single particle is available at a predefined position after a construction step. Another possibility would be to keep them in a bin in which they cannot stick to each other, i.e., the glue only really sticks when particles are on the workspace. To add a new particle to the workspace, for example, a~physical gate would have to be opened.

\subparagraph*{Constructibility.}
Beginning with a seed tile, a polyomino~$P$ consisting of $n$~tiles is constructible, if and only if there is a \emph{construction sequence} $\Sigma~=~(\sigma_1, \sigma_2, \dots, \sigma_{n-1})$ of~$n-1$ consecutive construction steps such that the resulting polyomino $P'$, induced by successively adding tiles with $\Sigma$, is identical to $P$. 
Reversing $\Sigma$ yields a \emph{deconstruction sequence} $\widetilde\Sigma$, i.e., a~sequence of deconstruction steps that iteratively removes tiles from $P$, eventually resulting in the empty workspace. We call a shape \emph{indestructible}, if there exist no deconstruction sequence for that shape.

\subparagraph*{Problem description.}
With the definitions given above, we are interested in the constructibility for a given shape. In particular, we consider the following decision problem.

\medskip
\noindent\textsc{Single Step Tilt Assembly Problem} (STAP)\newline
Given a shape $P$ of size $n$ and an anchored seed tile, does there exist a construction sequence of $n-1$ consecutive construction steps that constructs $P$?
\medskip

For non-constructible shapes we are interested in the maximum sized subshape that can be constructed. So we consider the following optimization problem.

\medskip
\noindent\textsc{Maximum Single Step Tilt Assembly Problem} (\textsc{Max}STAP)\newline
Given a non-constructible shape $P$ of size $n$ and an anchored seed tile, construct a subshape $P_{\max}\subset P$ where $\vert P_{\max}\vert$ is maximized.
\medskip

If the problems described are considered in a certain dimension, we specify this as a prefix of the respective problem name, e.g., the term ``$2$D-\problemtitle'' describes \problemtitle in the two-dimensional case.

Because we will use the fact that construction and deconstruction are basically the same problems throughout almost all arguments given in the paper, we want to restate the following result by Becker et al.~\cite{becker2018tilt} for the full tilt model.

\begin{theorem}[Theorem 2, \cite{becker2018tilt}]\label{lem:decomposability}
	A polyomino $P$ can be constructed if and only if it can be deconstructed using a sequence of tile removal steps that preserve connectivity. A construction sequence is a reversed deconstruction sequence.
\end{theorem}

To see that this is indeed true, consider a single construction step. The tile that is added to the assembly in this particular step can eventually be removed in a deconstruction of the shape on the exact same path. It is easy to see that this result can be adapted for the particular single step model considered in this paper.

Note that the definitions from above are given for the 2D setting. 
It is straightforward to extend these to the 3D setting, by letting $P\subset \mathbb{Z}^3$, introducing two additional directions in that a tile can move, and considering unit cubes instead of unit squares as tiles, and anchoring the workspace, that contains a seed tile at position~$(n,n,n)$, at position $(0,0,0)$. 

Furthermore, it is easy to see that \cref{lem:decomposability} applies in both models to three dimensions.

\section{3D-STAP is NP-complete}
Becker et al.~\cite{becker2018tilt} showed that it is NP-complete to decide whether or not a polycube is constructible in the full tilt model.
However, in our model, the paths that add new tiles to the assembly can be more complex than just a straight line, and thus we cannot simply adapt their proof.
In this section, we show that 3D-\problemtitle is NP-complete. The proof is based on a reduction from the NP-hard problem \textsc{Planar Monotone~3Sat}~\cite{moplsat}. This problem asks to decide whether a Boolean 3-CNF formula $\varphi$ is satisfiable, for which in each clause the literals are either all unnegated or all negated, and for which the clause-variable incidence-graph is planar.
Because of \cref{lem:decomposability}, we will argue that a polycube is deconstructible if and only if the corresponding Boolean formula $\varphi$ as an instance of planar monotone 3-sat is satisfiable.

\begin{figure}[ht]
	\centering
	\includegraphics[width=\textwidth]{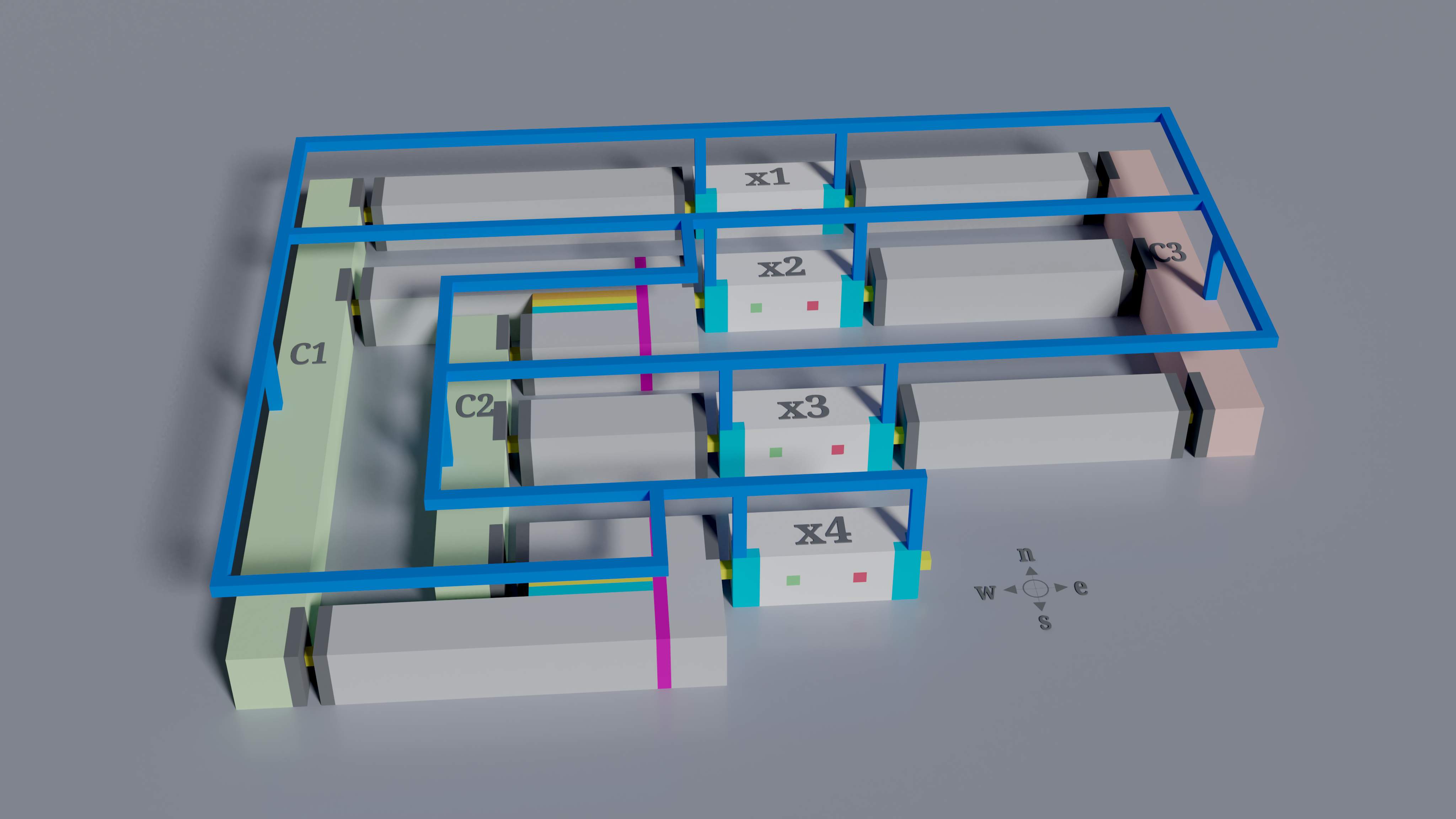}
	\caption{The figure gives a symbolic overview of the NP-hardness reduction. The reduction is based on the NP-hard problem \textsc{Planar Monotone 3Sat}. 
	The depicted instance is due to the Boolean formula $\varphi = (x_1\vee x_2 \vee x_4) \wedge (x_2\vee x_3 \vee x_4) \wedge (\overline{x_1}\vee \overline{x_2} \vee \overline{x_3})$. The variable gadgets are shown in white, while the unnegated and negated clause gadgets are shown in green and red, respectively. The gray cuboids represent connector gadgets. The `colorful lines' represent conjunction gadgets. 
	All~clauses and variables are connected by a blue frame above the construction to guarantee connectivity during a feasible deconstruction.}
	\label{fig:hardness-reduction}
\end{figure}

\subsection{Outline of the NP-hardness reduction} For every instance $\varphi$ of \textsc{Planar Monotone 3Sat}, we construct a polycube~$P_{\varphi}$ as an instance of 3D-\problemtitle. We consider a rectilinear planar embedding of the variable-clause incidence graph $G_{\varphi}$ of~$\varphi$ where the variable vertices are placed on a line, and clauses containing unnegated and negated literals are placed on either side, respectively. For a schematic overview of $P_{\varphi}$, consider \cref{fig:hardness-reduction}. In $P_{\varphi}$, each variable and each clause of~$\varphi$ is represented by a variable gadget and a clause gadget, respectively. To realize the fact that a variable can be contained in several clauses, we introduce a conjunction gadget. An edge in $G_{\varphi}$ is realized by a connector gadget. To guarantee connectivity during the deconstruction of the polycube~$P_{\varphi}$, we need to make sure that parts of the variable gadgets that are not participating in the satisfying assignment, are not disconnected in several parts. Therefore, we add a frame above the actual polycube, connecting all clauses with certain parts of the variable gadgets.

We can show that there is a deconstruction sequence for $P_{\varphi}$ if and only if~$\varphi$ is satisfiable. 
By using a checkered tile arrangement within all gadgets, we can enforce a specific deconstruction sequence.
On the one hand, we ensure that, due to the connectivity constraint, either the part of the variable that is connected to their unnegated or to their negated literal containing clauses can be deconstructed; thus, these deconstruction steps can be used to determine a valid variable assignment for $\varphi$.
This implies that, together with the conjunction gadgets, all clauses containing the respective literal can be deconstructed.
On the other hand, the other side of the variable gadget can only be deconstructed, if all other clauses are already deconstructed, i.e., if the clauses are satisfied by other variable assignments.

\subsection{Construction of the gadgets} In the following, we describe several polycubes that serve as gadgets in the NP-hardness reduction. In particular, we need polycubes for variables and clauses, as well as for (logic) conjunction. Furthermore, we need a gadget to connect several gadgets with each other, i.e., a connector gadget. 
All these gadgets are based on the following polycube that cannot be deconstructed if we restrict the deconstruction direction.

\subparagraph*{Indestructible wall.} A \emph{wall} is the polycube depicted in \cref{fig:indestructible1wall}. It consists of two layers, an odd 
sized \emph{solid layer}, and a checkered \emph{tooth layer}. This tooth layer consists of non-adjacent cubes (shown as dark gray cubes) at even positions. This construction can simply be modified to construct a \emph{k-wall}, see \cref{fig:adjacentwall} for examples. Note that $k$ can be at most 6, and there exist several $k$-walls for $k\in\{3,4\}$.

\begin{figure}
	\centering
	\begin{subfigure}[b]{0.33\textwidth}
		\centering
		\includegraphics[scale=0.032]{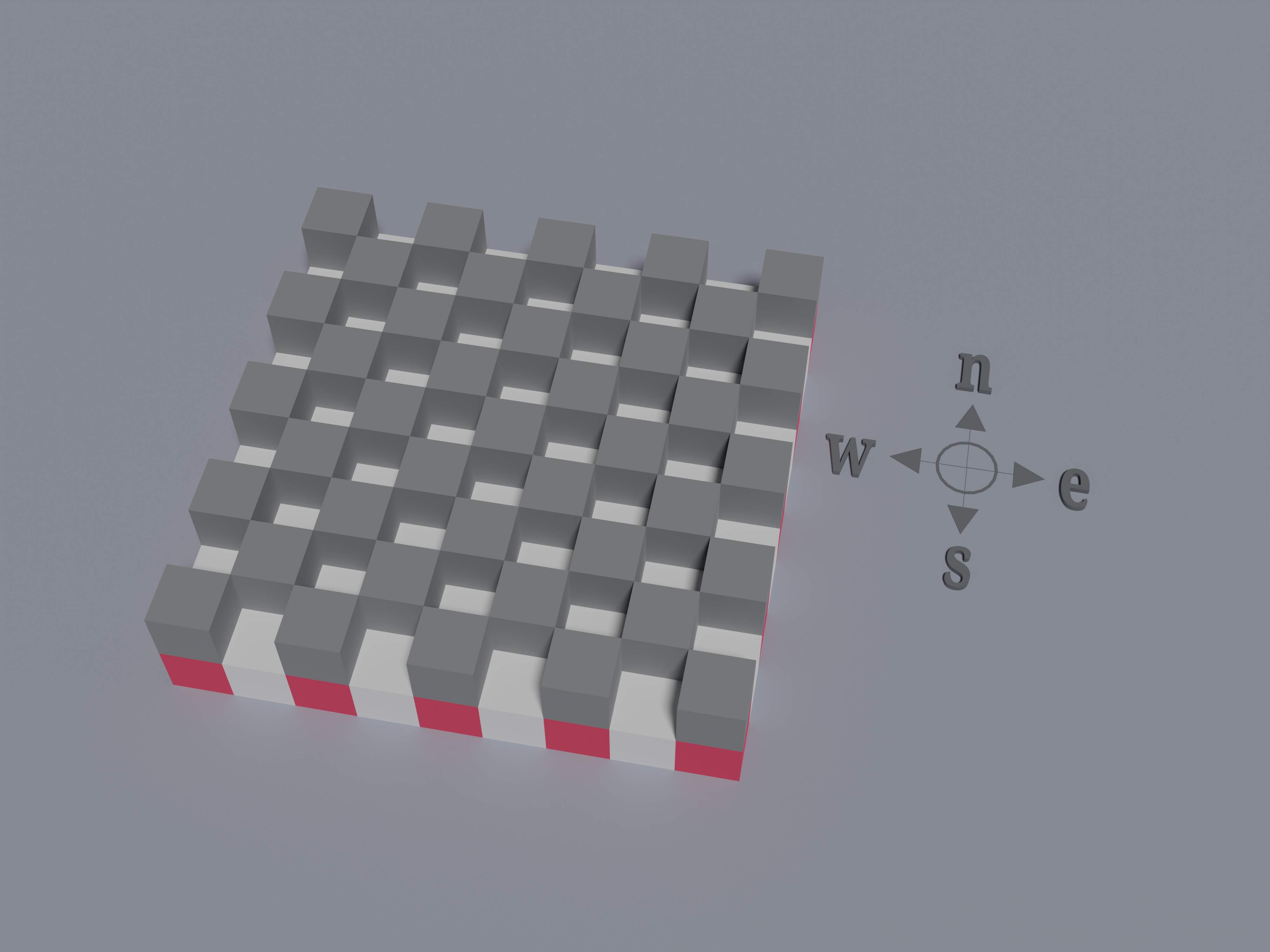}
		\caption{Indestructible 1-wall.}
		\label{fig:indestructible1wall}
	\end{subfigure}\hfill
	\begin{subfigure}[b]{0.33\textwidth}
		\centering
		\includegraphics[scale=0.032]{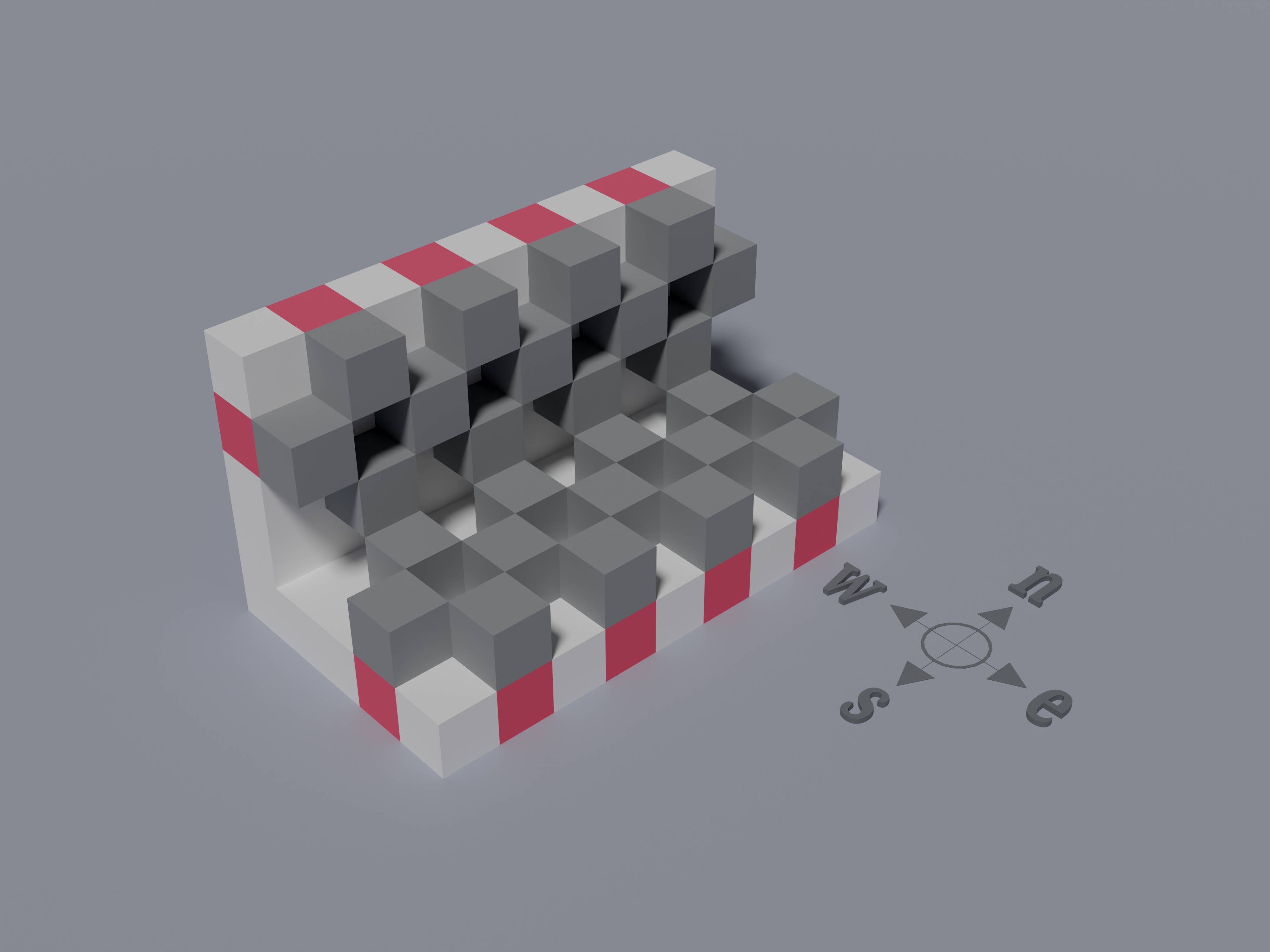}
		\caption{Indestructible 2-wall.}
		\label{fig:indestructible2wall}
	\end{subfigure}\hfill
	\begin{subfigure}[b]{0.33\textwidth}
		\centering
		\includegraphics[scale=0.032]{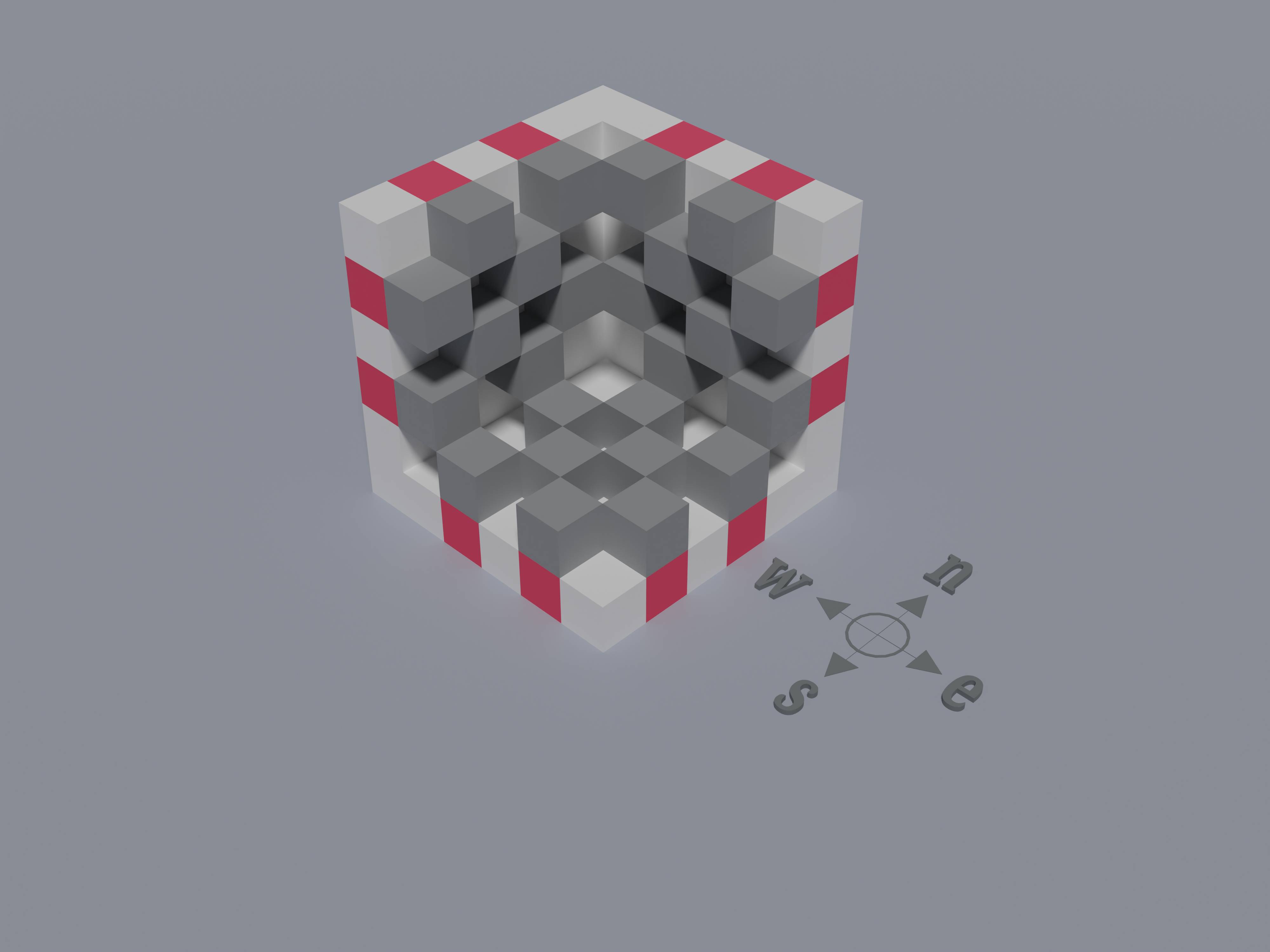}
		\caption{Indestructible 3-wall.}
		\label{fig:indestructible3wall}
	\end{subfigure}
	\caption{Indestructible walls. Red cubes indicate the respective positions of teeth, which in turn are shown in dark gray. Due to the design of these polycubes, a tooth must first be removed before the associated red cube can be removed from the assembly; otherwise the connectivity is lost.}
	\label{fig:adjacentwall}
\end{figure}

Recall that a deconstruction step is a sequence of moves that moves a tile of the polycube to position $(0,0,0)$. 
The crucial property of the wall is that it is not deconstructible \emph{from its solid layers}. First, we want to give an intuitive idea why this is true for the $1$-wall. For this, assume for a moment that the $1$-wall partitions the workspace into two parts, i.e., any tile in the top part of the workspace cannot pass the wall sideways to reach the bottom part, and vice versa. Furthermore, the solid layer and the position $(0,0,0)$ lay in the bottom part. It is clear that we need to remove a tile from the solid layer first. But no matter which one we choose, all its neighbors cannot be removed, because of the connectivity restriction. So, we cannot remove any tile from the checkered tooth layer.

With the assumption, that we only have access to the solid layers, we obtain the following.

\begin{lemma}\label{lem:indestructiblewall}
	A $k$-wall, $k\in\{1,\dots,6\}$, is not deconstructible from its solid layers.
\end{lemma}

\begin{proof}
	Suppose for the sake of a contradiction that a $k$-wall is deconstructible from its solid layers, and  
	let $\widetilde\Sigma=(\sigma_1, \dots, \sigma_n)$ be a deconstruction sequence.
	Because we only have access to the solid layers, $\sigma_1$ clearly has to remove a cube from a solid layer. 
	We partition the solid layer in cubes at even and odd positions. 
	Because there are cubes in the tooth layer at even positions, we can only remove cubes at odd positions from the solid layer to maintain connectivity. 
	But then there cannot be a $\widetilde\sigma_i\in \widetilde\Sigma$ such that $\widetilde\sigma_i$ removes a cube from the tooth layer. 
	This is a contradiction to the existence of $\widetilde\Sigma$. 
	Thus, a wall is not deconstructible from the solid layer.
\end{proof}

We show that the specific design of the tooth layers of a $k$-wall is necessary to guarantee the indestructibility, i.e., if at least one tooth is missing, the resulting polycube is deconstructible from the solid layer.

\begin{lemma}\label{lem:decomposablewall}
	There is at least one position $p$ at a tooth layer of a $k$-wall, $k\in\{1,\dots,6\}$ such that the $k$-wall is deconstructible from the solid layer if $p$ is free.
\end{lemma}

\begin{proof}
	Let $p$ be the position of the missing tooth, and consider the respective position in the solid layer. 
	Removing this tile and the four adjacent tiles will yield a hole that is large enough that tiles can pass through it. Thus, tiles from the tooth layer can be removed. 
	To do so, we position a tile from the tooth layer above $p$ and move it in the direction of the hole such that there will be no face contact to other tiles. 
	Thus, all teeth can be removed one by one, followed by deconstructing the remaining parts of the solid layer.
\end{proof}

A simple observation is that these $k$-walls can arbitrarily be enlarged without losing the property of being indestructible. Of particular interest for the reduction are enlarged $6$-walls, called \emph{cuboids}, that will serve as \emph{clause} and \emph{connector gadgets}.

Another crucial observation is that because a $k$-wall is not deconstructible from its solid layers, we can leave out several cubes of the solid layers so that the remaining shape is disconnected into two parts, see \cref{fig:disconnected2wall} for an example of a disconnected $2$-wall. This insight will lead to a configuration that allows for a decision, i.e., that will serve as the variable gadget.

\begin{figure}
	\centering
	\includegraphics[scale=0.032]{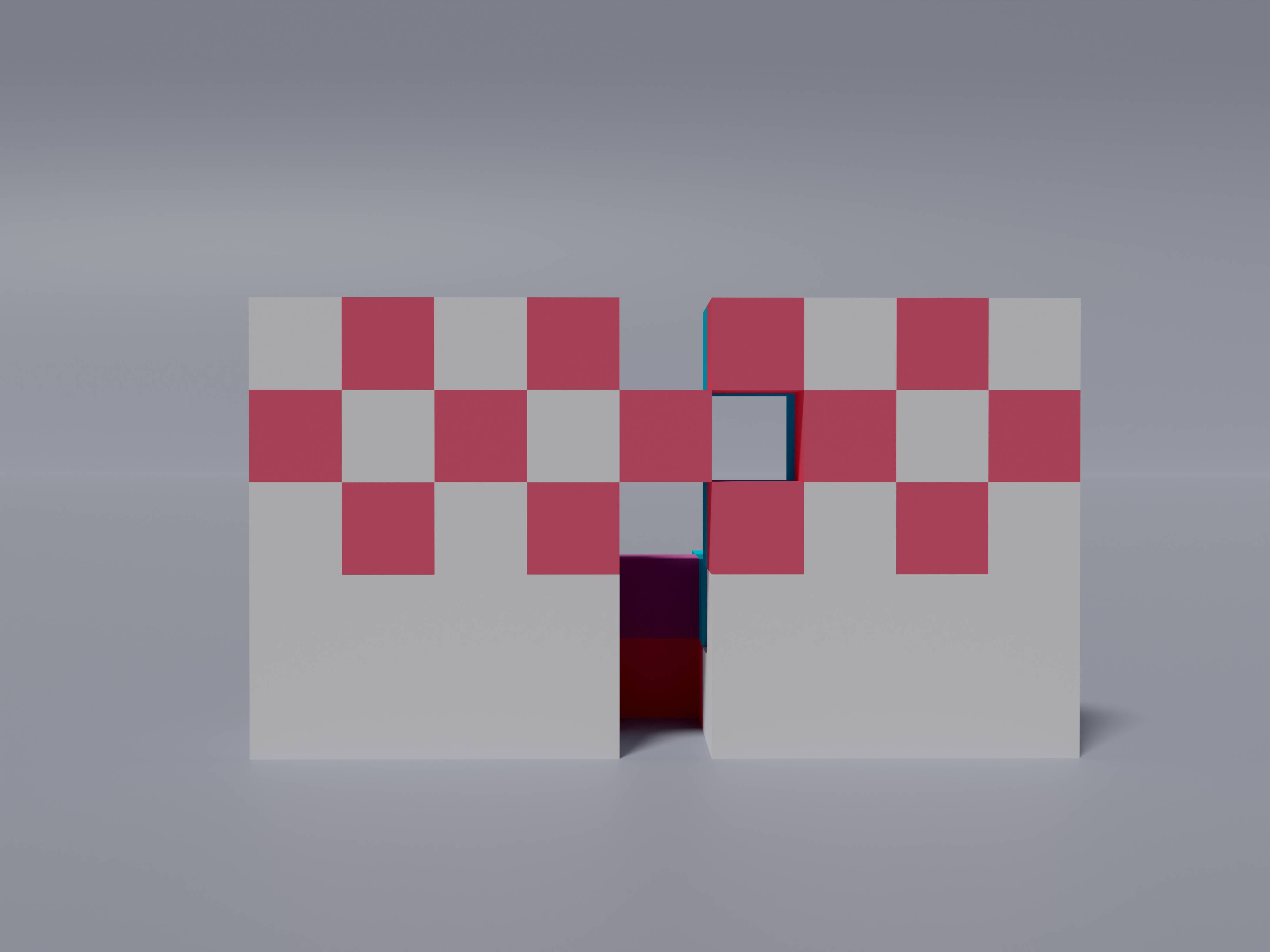}\hfill
	\includegraphics[scale=0.032]{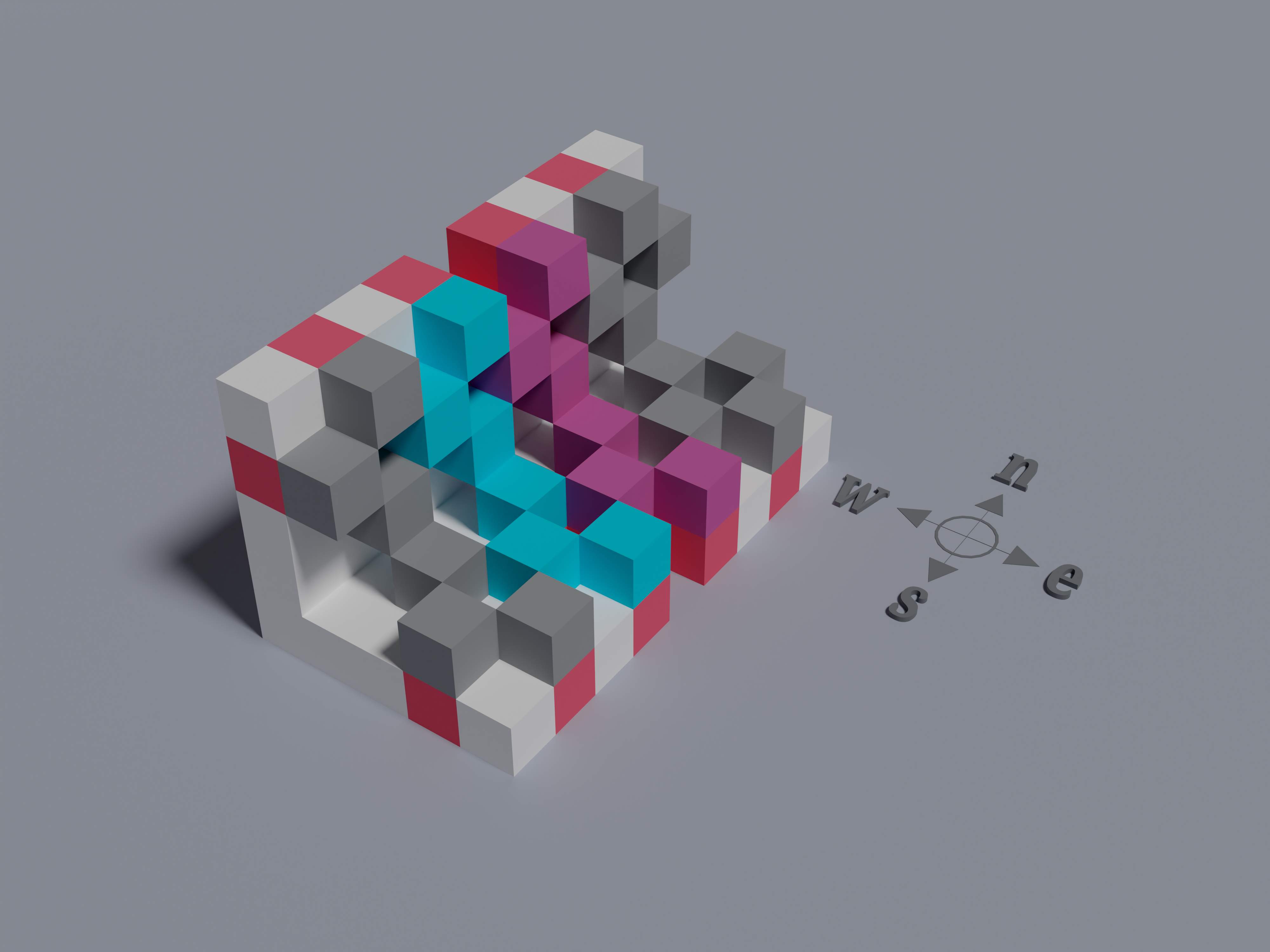}\hfill
	\includegraphics[scale=0.032]{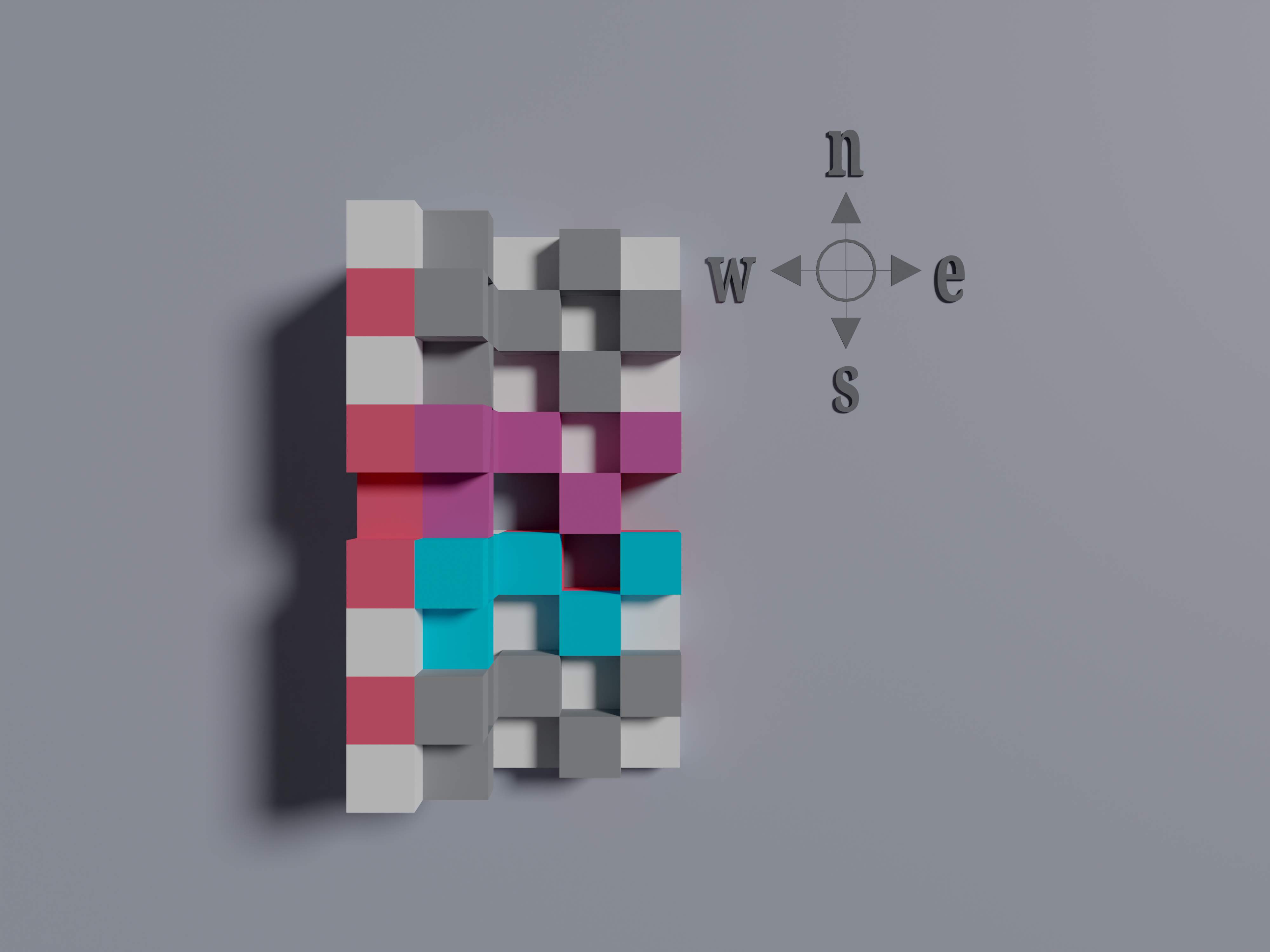}
	\caption{Different views on a disconnected 2-wall. To separate the 2-wall into two disconnected components, light gray cubes of the solid layer are removed in the space between the blue and purple teeth.}
	\label{fig:disconnected2wall}
\end{figure}

If two cuboids have to be connected, we place them at distance one to each other and add a single cube to connect them. Furthermore, we remove all cubes in a $3\times 3$ area at matching sides such that we can move cubes from the inside of one cuboid to the other through these holes, see \cref{fig:variablegadget2,fig:variablegadget3} for illustration.

\subparagraph*{Variable gadget.} The variable gadget consists of two indestructible cuboids ($Q_1$ and $Q_2$) that share a solid layer, see \cref{fig:variablegadget3} for an exploded illustration.
As shown in \cref{fig:variablegadget}, we remove tiles (similar to \cref{fig:disconnected2wall}) to separate an L-shaped part of each cuboid (light blue tiles). These shapes are then reconnected by two bridges (green and orange tiles), see \cref{fig:variablegadget0}.
Additionally, the L-shaped parts are connected by a thin frame above the cuboids (dark blue tiles).

\begin{figure}[ht]
	\centering
	\begin{subfigure}[b]{0.48\textwidth}
		\centering
		\includegraphics[scale=0.0475]{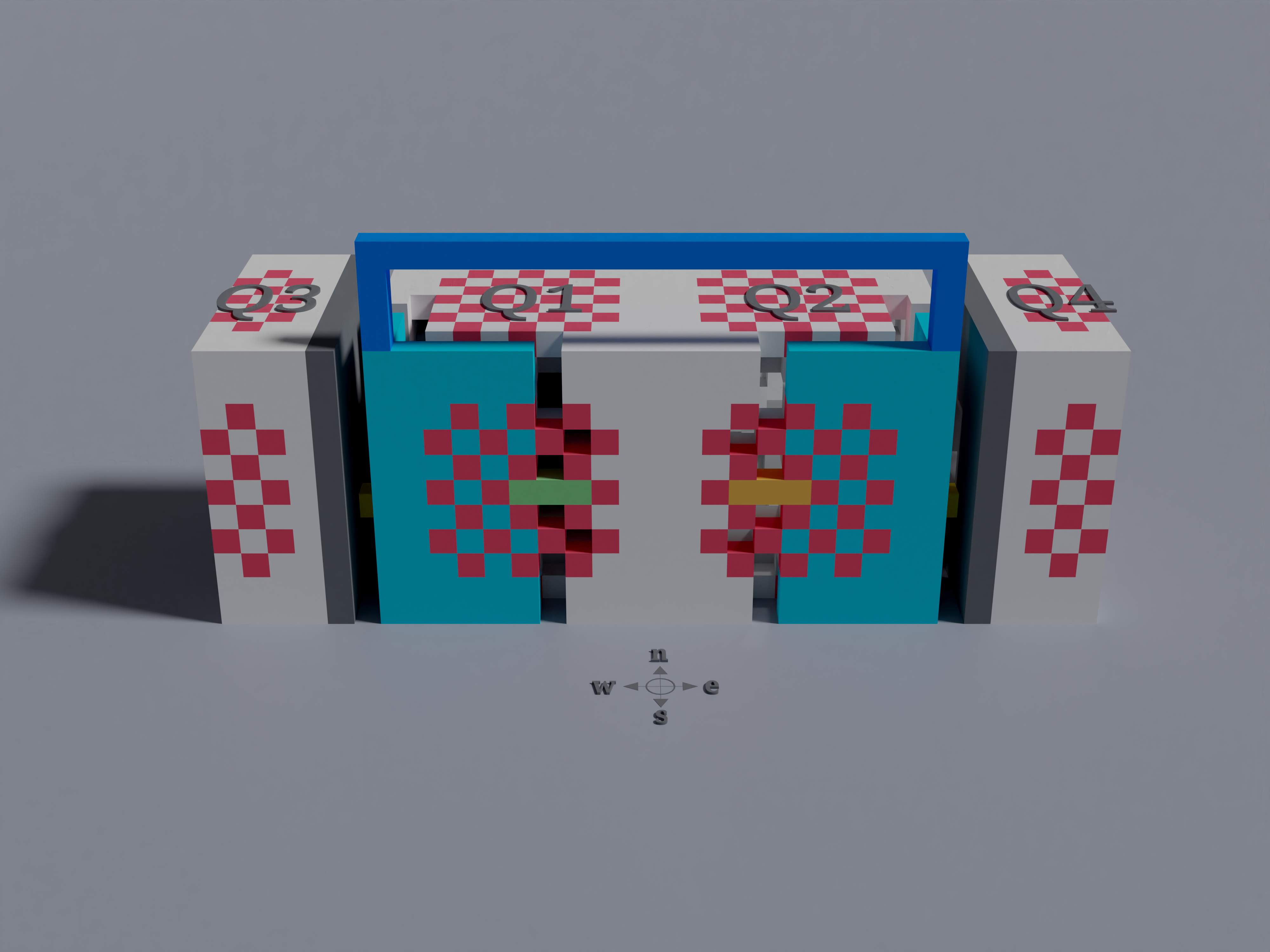}
		\caption{Front view}
		\label{fig:variablegadget0}
	\end{subfigure}\hfill
	\begin{subfigure}[b]{0.48\textwidth}
		\centering
		\includegraphics[scale=0.0475]{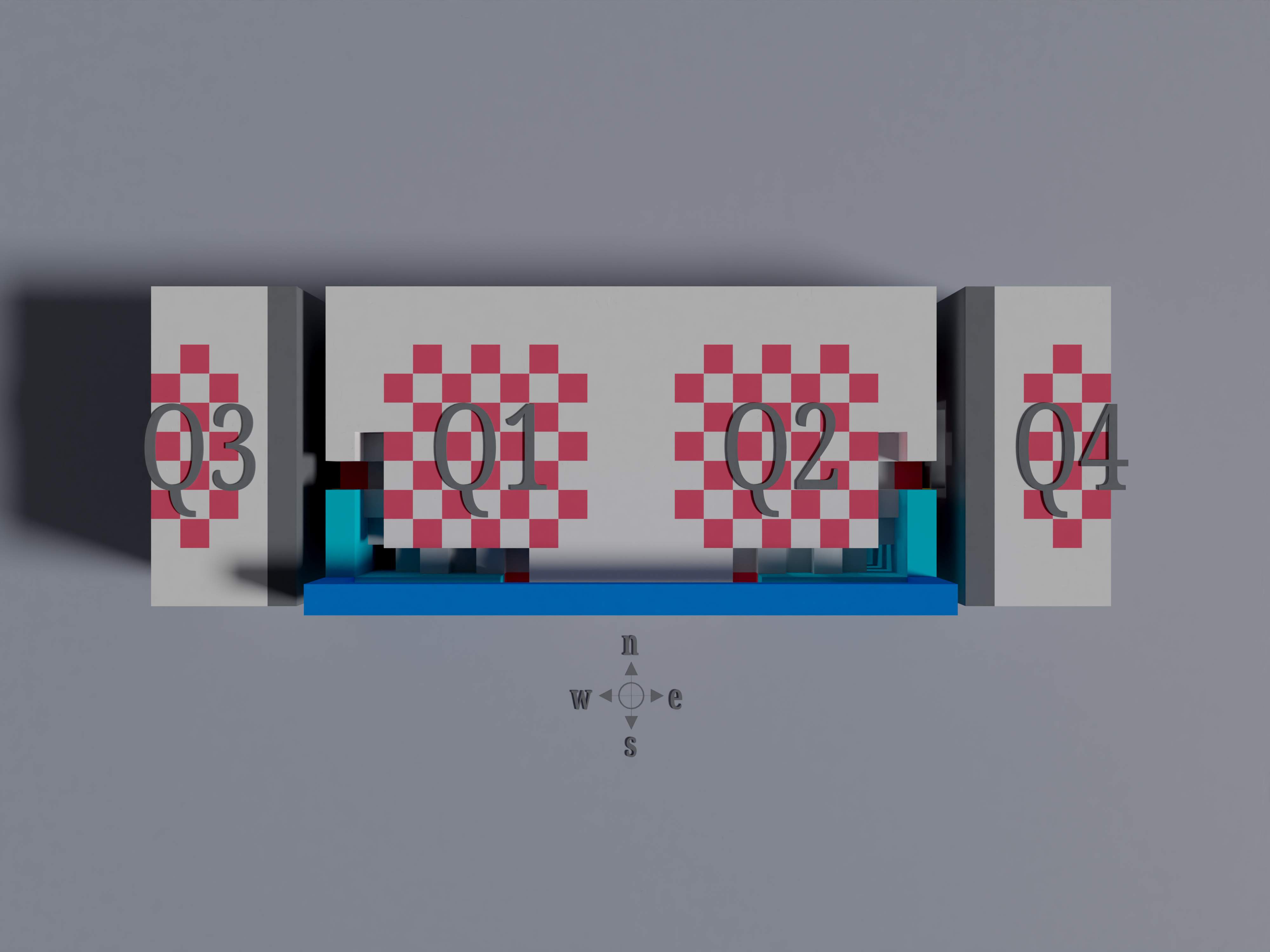}
		\caption{Top view}
		\label{fig:variablegadget1}
	\end{subfigure}\vspace{0.3cm}
	\begin{subfigure}[b]{0.48\textwidth}
		\centering
		\includegraphics[scale=0.0475]{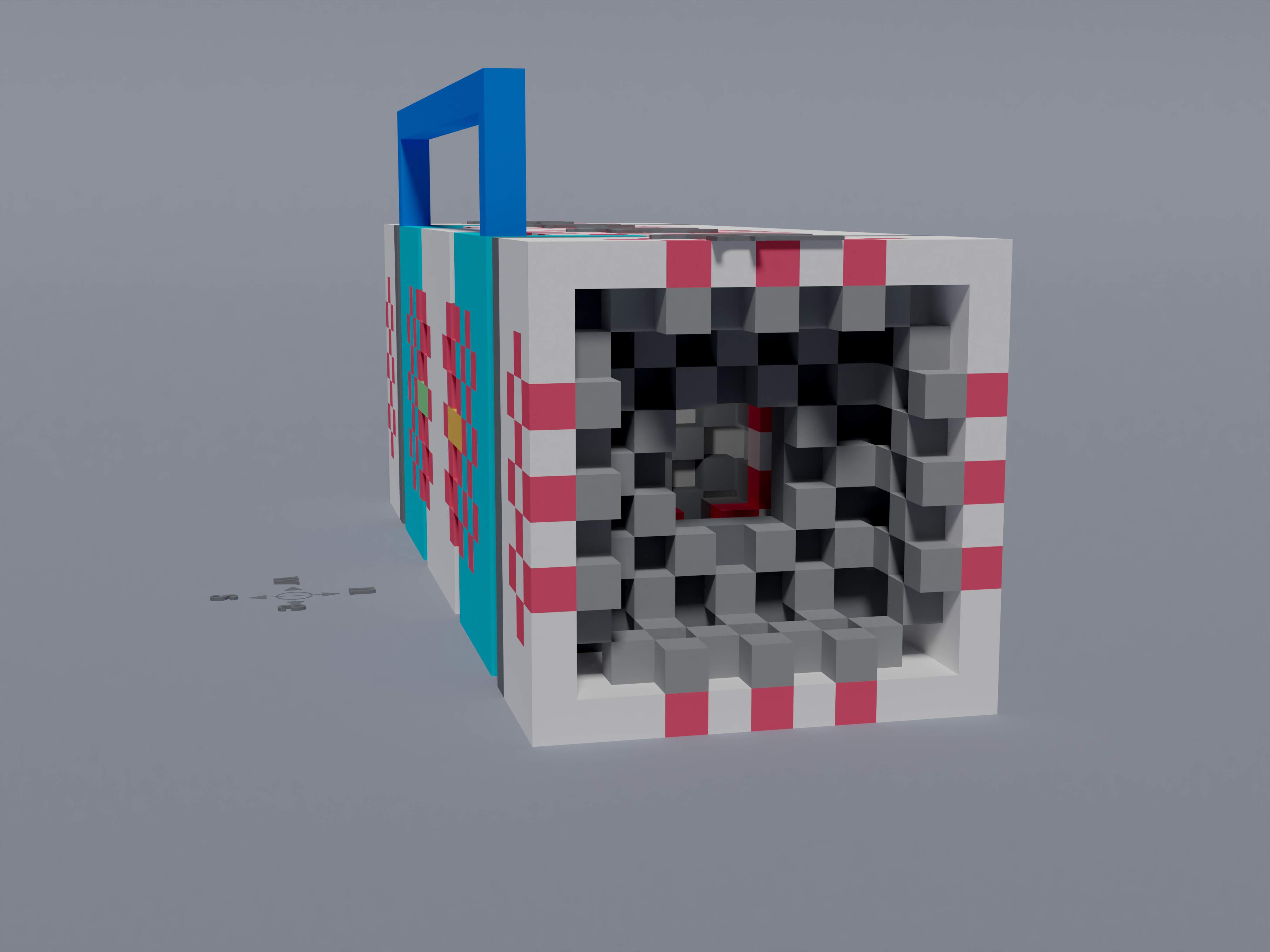}
		\caption{Side view}
		\label{fig:variablegadget2}
	\end{subfigure}\hfill
	\begin{subfigure}[b]{.48\textwidth}
		\centering
		\includegraphics[scale=0.0475]{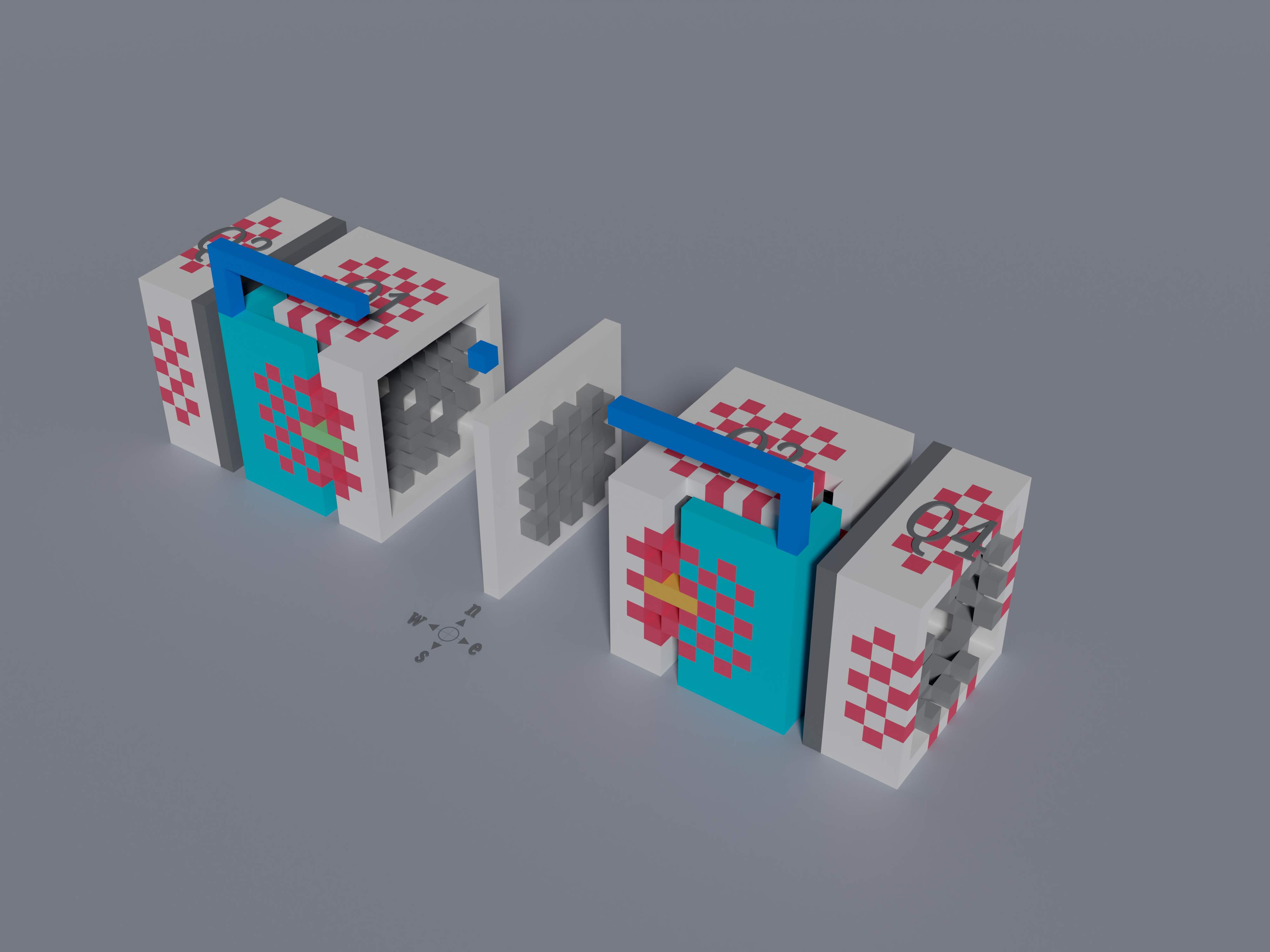}
		\caption{Exploded view}
		\label{fig:variablegadget3}
	\end{subfigure}
	\caption{Different views on the variable gadget. The actual variable gadget consists of the indestructible cuboids $Q_1$ and $Q_2$, that are modified in such a way that there are L-shaped parts (shown in light blue) that are only connected to the remaining assembly via the green and orange bridges. The dark blue tiles indicate a part of the connectivity frame.}
	\label{fig:variablegadget}
\end{figure}

\begin{observation}\label{lem:variablegadget}
	Solely removing the green and orange tiles of a variable gadget results in a disconnected shape.
\end{observation}

As a consequence of \cref{lem:variablegadget}, the forced choice of removing either the green or the orange tiles, can be used to determine an assignment for the respective Boolean variable. It remains to show how a variable gadget can be deconstructed, if additional cuboids are attached at each side.

\begin{lemma}\label{lem:variablegadgetdestruction}
	Let $P$ be a polycube that is put together by a variable gadget and one cuboid ($Q_3$ and $Q_4$) at each end, connected to the respective L-shaped parts. Then $P$ is only deconstructible if at least $Q_3$ or $Q_4$ is deconstructible.
\end{lemma}

\begin{proof}
	Without loss of generality, let $Q_3$ be deconstructible. Then the orange tiles can be removed in order to deconstruct $Q_4$ (by \cref{lem:decomposablewall}), and $Q_2$ afterwards. Due to the hole between $Q_1$ and $Q_3$ and the assumption that $Q_3$ is deconstructible, $Q_1$ can also be deconstructed. 
	
	On the other hand, if neither $Q_3$ nor $Q_4$ is deconstructible, only the green or the orange tiles can be removed (by \cref{lem:variablegadget}). But then, either $Q_3$ or $Q_4$ can be deconstructed, but not both, resulting in an indestructible shape.
\end{proof}

As the last ingredient for our NP-hardness reduction we need a gadget that realizes a conjunction.
This gadget will be used to guarantee that a variable gadget can be completely deconstructed if and only if all clauses in which the respective variable participates are satisfied.

\subparagraph*{Conjunction gadget.}  As illustrated in \cref{fig:conjunctiongadget}, the conjunction gadget is T-shaped. 
The wall between the cuboids $Q_1$ and $Q_2$ contains teeth to both sides, whereas the wall at cuboid $Q_3$ has teeth except for the positions where the T-shape is connected. This connection will be the crucial part to deconstruct this gadget.
At all three positions ($Q_1$, $Q_2$, and $Q_3$) we attach connector gadgets leading either to another conjunction, to a variable, or to a clause gadget. Note that these connector gadgets have the same size as the conjunction gadget, i.e., the solid layers of the connectors and the conjunction gadget match.

\begin{figure}
		\centering
		\includegraphics[scale=0.032]{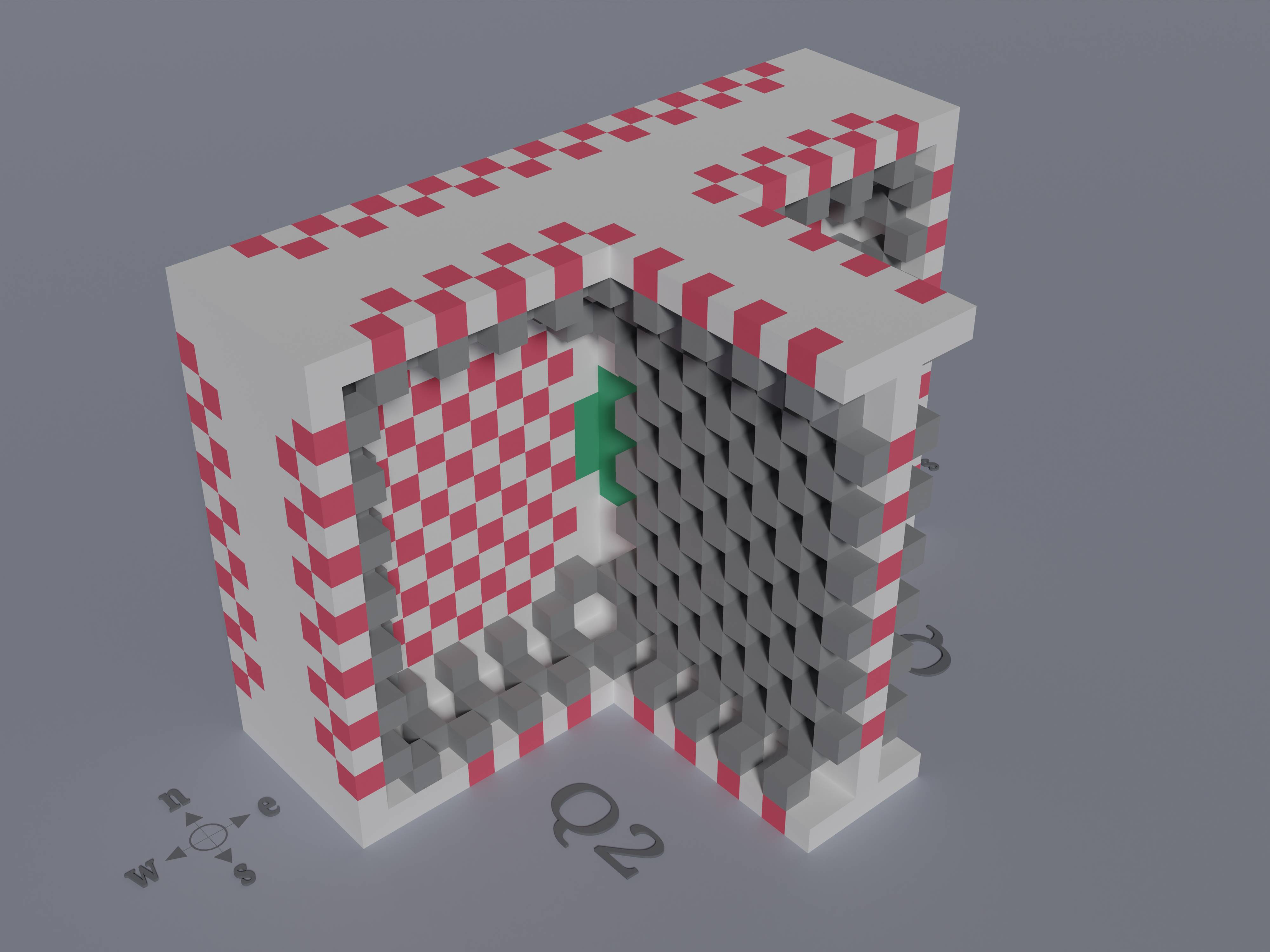}\hfill
		\includegraphics[scale=0.032]{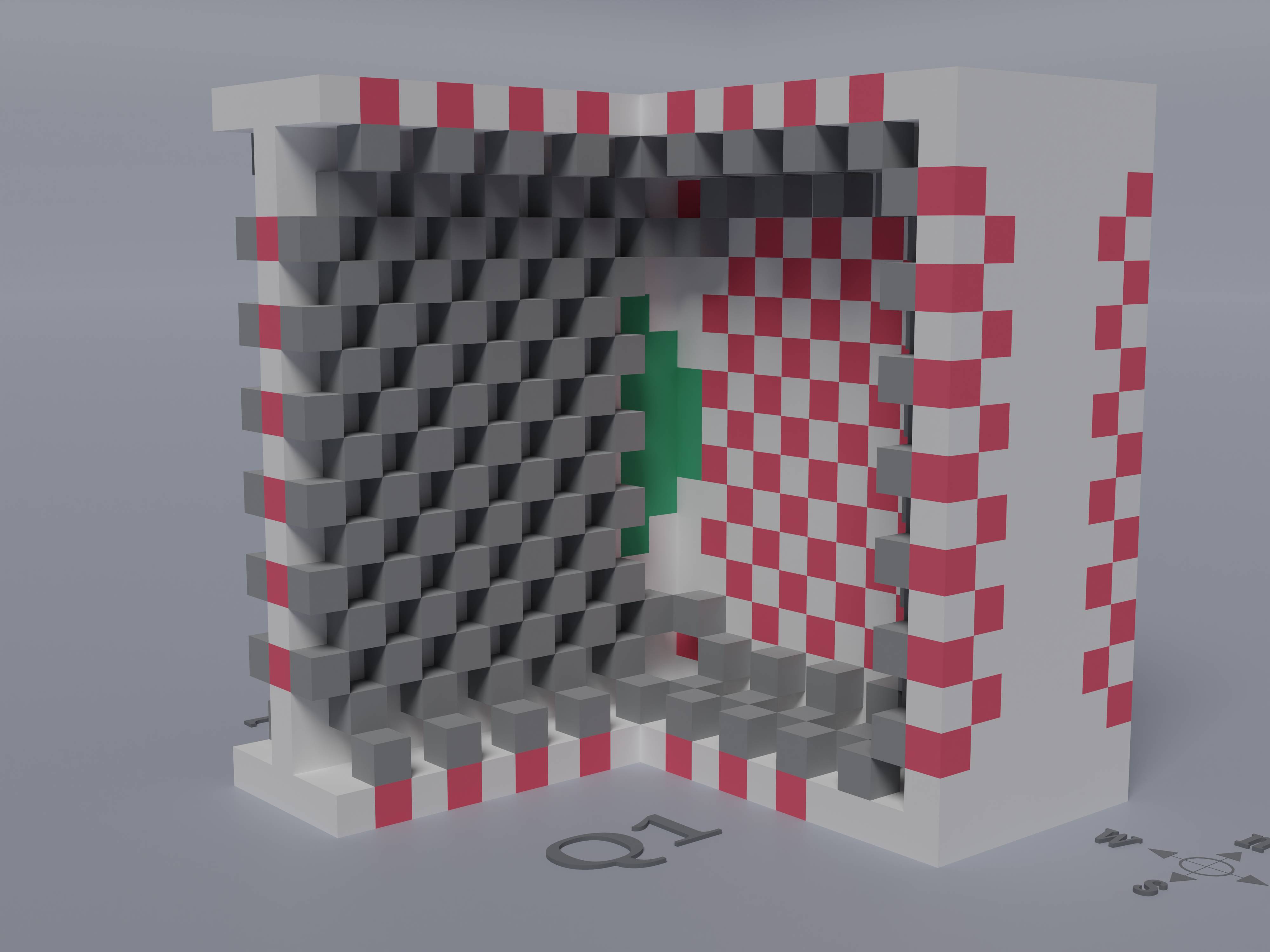}\hfill
		\includegraphics[scale=0.032]{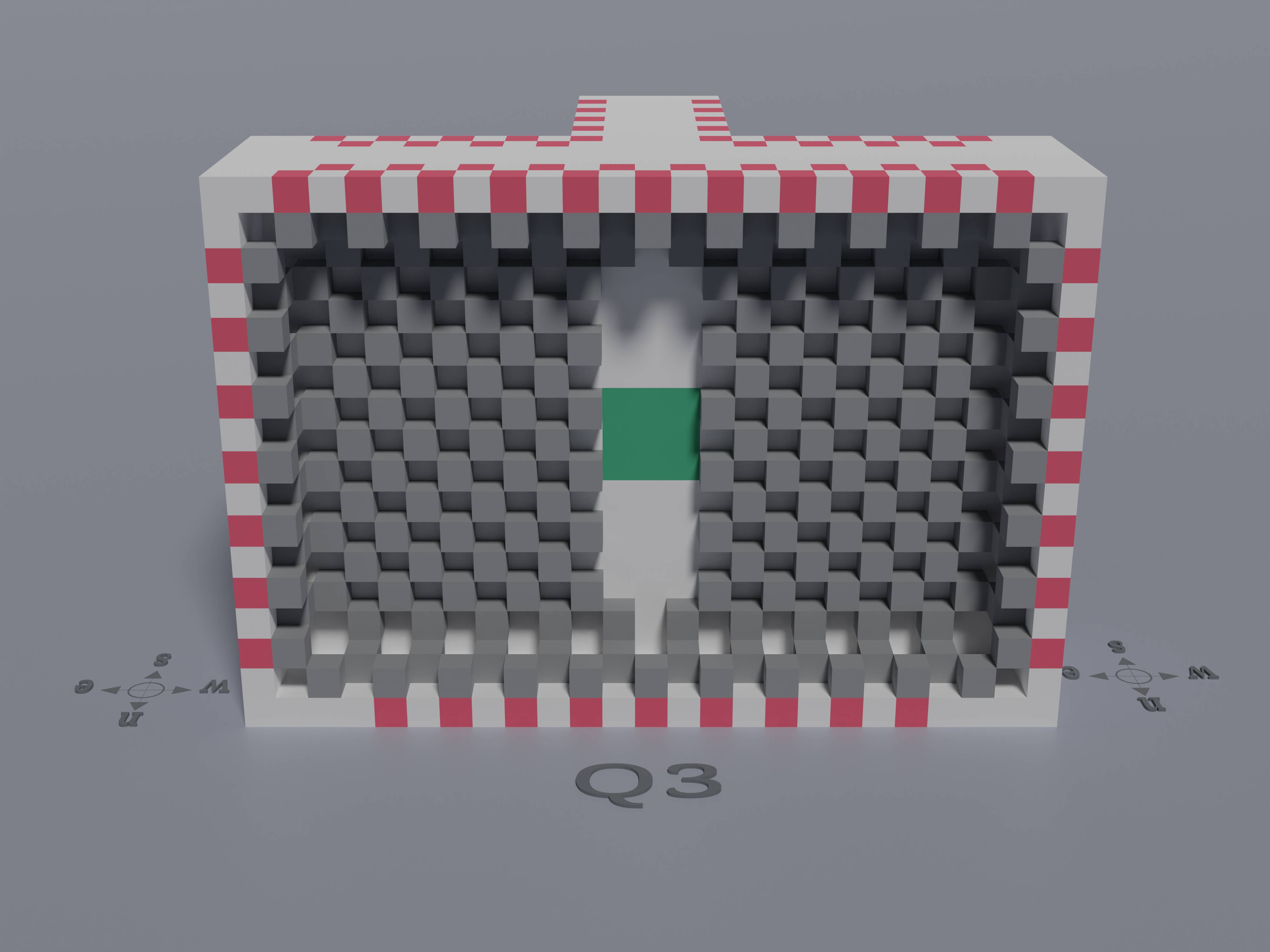}
	\caption{Different views on the conjunction gadget. At the indicated positions, three indestructible cuboids $Q_i$ are attached. The whole construction is deconstructible, if and only if both, $Q_1$ and $Q_2$, or $Q_3$ is deconstructible. This follows from the design of the respective teeth layers, as well as the dark green tiles.}
	\label{fig:conjunctiongadget}
\end{figure}

We can show that this gadget is deconstructible if and only if the cuboid at $Q_3$, or both cuboids at $Q_1$ and $Q_2$ are deconstructible.

\begin{lemma}\label{lem:conjunctiongadget}
	Let $P$ be a polycube that is put together by three cuboids $Q_1, Q_2,$ and $Q_3$ which are connected by a conjunction gadget. Then $P$ is deconstructible if and only if $Q_1$ and $Q_2$ are both deconstructible, or $Q_3$ is deconstructible.
\end{lemma}

\begin{proof}
	For the proof, we distinguish two cases.
	
	First, let $Q_1$ and $Q_2$ be deconstructible. Thus, all teeth from the respective insides can be removed. After this, the wall that is shared by $Q_1$ and $Q_2$ becomes deconstructible. Afterwards, by removing the dark green tiles (see~\cref{fig:conjunctiongadget}), we create a hole to reach the inside of $Q_3$ that makes $Q_3$ deconstructible as well. Thus, $P$ is deconstructible.
	
	Now assume that either $Q_1$ or $Q_2$ is deconstructible, but not both. Without loss of generality consider $Q_1$ to be deconstructible. Then all teeth from its inside can be removed. By this, neither $Q_2$ nor $Q_3$ become deconstructible because the teeth on the respective sides do not permit the removal of enough tiles such that tiles from the respective insides can pass through. Thus, $P$ is not deconstructible.
	
	It remains to show that $P$ is deconstructible if $Q_3$ is deconstructible. Because $Q_3$ is deconstructible, we can remove the teeth at this side of the conjunction gadget. Because there are no teeth at the opposite sides, the whole wall can be deconstructed. This results in holes to the inside of $Q_1$ and $Q_2$ such that these cuboids become deconstructible as well.
\end{proof}

By putting all these together, we obtain the following.

\begin{theorem}\label{thm:hardness}
	\textsc{3D-}\problemtitle is NP-complete.
\end{theorem}

\begin{proof}
	First we note that the problem of deciding whether or not a polycube is constructible is in NP. For this, we guess a permutation of the involved tiles and check if there is a feasible construction step for each tile. It is easy to see that a single construction step has length~$O(n)$. Because the target polycube has size $n$, a construction sequence has length $O(n^2)$.
	
	To show NP-hardness, consider a rectilinear planar embedding of the variable-clause incidence graph $G_{\varphi}$ of a given \textsc{Planar Monotone 3Sat} formula $\varphi$, where the variable vertices are placed horizontally in a row, and clauses containing unnegated and negated literals are placed above and below this row, respectively. 
	We place a variable gadget for every variable (white in \cref{fig:hardness-reduction}), and a cuboid for each clause (green in \cref{fig:hardness-reduction}). 
	These blocks are connected via connector cuboids (gray) and conjunction gadgets whenever the number of unnegated or negated occurrences of a variable is larger than one.
	On top of this construction, we use a frame to hold all clauses and L-shaped parts in the variable gadgets together. This will be necessary to deconstruct the shape completely whenever the underlying Boolean 3Sat formula is satisfiable.
	
	\begin{claim}
		If there is a deconstruction sequence $\widetilde\Sigma$ for $P_{\varphi}$, then there is a satisfying assignment for $\varphi$.
	\end{claim}
	
	In order to deconstruct a clause gadget, at least (parts of) one of its literal containing variable gadgets has to be already deconstructed. Thus, every deconstruction sequence has to begin with deconstruction steps that remove either the green or orange tiles from any variable gadget. As argued in \cref{lem:variablegadget}, the green and orange tiles cannot both be removed, i.e., setting a variable to true and false simultaneously is not possible. Because $\widetilde\Sigma$ is a deconstruction sequence, eventually each clause will be deconstructed.
	
	Therefore, there is a satisfying assignment for $\varphi$ given by the order in which each variable of $P_{\varphi}$ is deconstructed by $\widetilde\Sigma$. In particular, for each $x_i \in \varphi$, we set $x_i = 1$ if the respective green tiles are removed first, and $x_i = 0$ otherwise.
	
	\begin{claim}
		If the Boolean formula $\varphi$ is satisfiable, then $P_{\varphi}$ is deconstructible by some deconstruction sequence $\widetilde\Sigma$.
	\end{claim}
	
	Let $\alpha$ be a satisfying assignment of $\varphi$. Then the polycube $P_{\varphi}$ can be deconstructed as follows: According to whether a variable is set to true or false in $\alpha$, either the green or orange tiles are removed, respectively. These deconstruction steps produce holes, large enough for a deconstruction of the whole literal representing cube, as argued in \cref{lem:variablegadget}. Afterwards, because there is a hole from the clause cuboids to the literals, all clauses that contain these literals can be deconstructed by \cref{lem:decomposablewall}. Because $\alpha$ is a satisfying assignment for $\varphi$, all clause gadgets and the respective parts of the variable gadgets can be deconstructed. The respective parts of the variable gadgets that are not participating in the satisfying assignment are connected by the overlaying connectivity frame. Because the clauses are already deconstructed, there are holes through which the remaining parts of the variable gadgets can be deconstructed. Removing the connectivity frame results in deconstructing~$P_{\varphi}$.
	
	These two claims complete the proof.
\end{proof}

\section{Optimization variant and approximation}

For polyominoes and polycubes that cannot be constructed, it is natural to consider the problem of constructing a subshape of maximum size. 
We show that for each shape $P$ of size~$n$ in dimension $d$, a portion of $\Omega(n^{\nicefrac{(d-1)}{d}})$ can always be constructed, implying an $\Omega(n^{\nicefrac{-1}{d}})$-approximation for \textsc{MaxSTAP}.

\begin{figure}[h]
	\begin{subfigure}{0.33\textwidth}
		\includegraphics[page=1, scale=0.85]{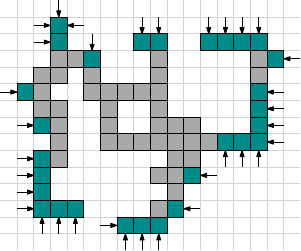}
		\caption{}
		\label{fig:approx1}
	\end{subfigure}\hfill
	\begin{subfigure}{0.33\textwidth}
		\includegraphics[page=2, scale=0.85]{figures/figs-final}
		\caption{}
		\label{fig:approx2}
	\end{subfigure}\hfill
	\begin{subfigure}{0.33\textwidth}
		\includegraphics[page=3, scale=0.85]{figures/figs-final}
		\caption{}
		\label{fig:approx3}
	\end{subfigure}
	\caption{(a) Boundary tiles (dark cyan) and non-boundary tiles (gray). (b) Every step of the  algorithm adds a tile in two steps: (1)~Move a new tile $t$ to a boundary position $p$ of $P$ that is free. (2)~Move $t$ on $P$ from $p$ to a position that is adjacent to a tile of the current polyomino $P'$. (c) The curve $B$ (blue), and the set $T$ (orange).}
	\label{fig:approx}
\end{figure}

\begin{theorem}\label{thm:constructability}
	In dimension $d=2, 3$, the greedy algorithm is an $\Omega(n^{\nicefrac{-1}{d}})$-approximation for \textsc{MaxSTAP}.
\end{theorem}

\begin{proof}
	We prove the theorem by showing that greedily filling up accessible free positions leads to a polyomino $P' \subseteq P$ with $\nicefrac{\vert P'\vert}{\vert P\vert} \in \Omega(n^{\nicefrac{-1}{d}})$. 
	A position $p$ is \emph{accessible} with respect to a polyomino $P$ if we can move a tile $t$ to this position such that $t$~is never adjacent to a tile of $P$ unless $t$ lies on $p$. A tile $p \in P$ is a \emph{boundary tile} of $P$ if $p$ is accessible with respect to $P \setminus p$, see \cref{fig:approx1}. 
	If there is a boundary tile $p$ of $P$ that is not part of $P'$, we add a new tile $t$ to $P'$ in two steps, see \cref{fig:approx2}: 
	(1)~We move $t$ to the position $p$. 
	(2)~We move $t$ on $P$ to a position adjacent to a tile of $P'$. This implies that the greedy algorithm ends up with a polyomino holding all boundary tiles of $P$.
	
	Next, we show a lower bound on the boundary tiles of $P$ of $\Omega(n^{\nicefrac{(d-1)}{d}})$.
	We will only show this for dimension $d=2$, 
	similar 
	arguments hold for $d=3$. 
	Let $B$~be the union of all edges lying between an accessible and a non-accessible position with respect to $P$, see blue curve in \cref{fig:approx3}.
	$B$ is a non-self-intersecting curve by the definition of accessible positions. 
	Thus, $B$ partitions the plane into a bounded area $A$ containing $P$ and an unbounded area. 
	Let $T$ be the union of all positions from $A$ sharing at least a corner with $B$, see the orange positions in \cref{fig:approx3}. 
	Then $\vert T\vert \geq \sqrt{\vert A\vert}$. 
	Note that not each position of $T$ is occupied by~$P$, see the light gray positions in \cref{fig:approx3}. 
	Let $T'$ be the positions along $T$ that share an edge with $B$. It is easy to see that $2\vert T'\vert\geq \vert T\vert$.
	Each position $p$ from $T'$ that is not a boundary tile from $P$ is adjacent to a boundary tile $p'$ from $P$.
	We call $p'$ a \emph{blocking tile} of $p$. 
	Each boundary tile is a blocking tile for a constant number of positions $p \in T'$ that are not a boundary tile. 
	Hence, there are $\Omega(\vert T'\vert)=\Omega(\vert T\vert)$ many boundary tiles. 
	Because $P \subseteq A$, we obtain $\vert P\vert \leq \vert A\vert$ implying $\vert T\vert \in \Omega(\sqrt{\vert A\vert})\subseteq\Omega(\sqrt{\vert P\vert})$.
	Hence, the shape $P'$ constructed by the greedy algorithm has at least $\Omega(n^{\nicefrac{1}{2}})$ boundary tiles. A similar surface-to-volume argument implies a lower bound of $\Omega(n^{\nicefrac{2}{3}})$ in the three-dimensional case. 
	
	Therefore, in dimension $d= 2,3$, this approach yields an approximation factor of $\nicefrac {\vert P'\vert}{\vert P\vert} = \Omega(n^{\nicefrac {-1} d})$.
\end{proof}
\section{Efficient algorithms for special classes of shapes}
Due to the equivalence of construction and deconstruction, we pursue the goal of classifying tiles that can be removed from a shape, without harming its deconstructibility. In general, this seems to be really difficult. On the one hand, it is not sufficient to restrict the search for removable tiles to \emph{corners} (tiles with exactly one horizontal and one vertical neighbor), because for successfully deconstructing a polyomino, it may be necessary to remove non-corner tiles first, see \cref{fig:problempolyomino1}.
On the other hand, removing non-corner tiles can result in an indestructible subshape, again see \cref{fig:problempolyomino1}.
Furthermore, even in simple polyominoes, the removal of a corner tile can result in an indestructible subshape, see \cref{fig:problempolyomino2}.
Note that the latter is not the case in the full tilt model~\cite{becker2018tilt}.

\begin{figure}[h]
	\begin{subfigure}[b]{0.5\textwidth}
		\centering
		\includegraphics[page=11]{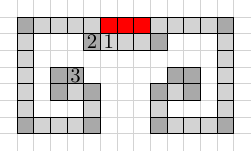}
		\caption{}
		\label{fig:problempolyomino1}
	\end{subfigure}\hfill
	\begin{subfigure}[b]{0.5\textwidth}
		\centering
		\includegraphics[page=12]{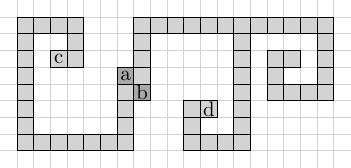}
		\caption{}
		\label{fig:problempolyomino2}
	\end{subfigure}
	\caption{(a) No corner tile (dark gray) can be removed, because they either do not have a deconstruction step or are essential for connectivity. 
		Successively removing the tiles~1, 2, and 3 by suitable deconstruction steps, the obtained  shape can easily be deconstructed.
		Removing the red tiles first results in an indestructible shape.
		(b) By removing $a$ first, we can remove the spiral starting with $c$ and afterwards the rest. By removing $b$ first, the spiral starting with $d$ can be removed, but the remaining shape is indestructible.}
	\label{fig:problempolyomino}
\end{figure}

Therefore we consider special classes of shapes and show that the problem becomes easy in the case of trees and scaled shapes.

\subsection{Tree shapes}
We show that \problemtitle can be decided in linear time (in the size of the shape) for the class of tree-shapes by a greedy algorithm. We show this in detail for the two-dimensional setting, i.e., for tree-shaped polyominoes. The results can simply be adapted for tree-shaped polycubes in three dimensions. Because the removal of a tile with more than one neighbor results in splitting the polyomino in several parts, we are restricted to remove tiles with exactly one neighbor, i.e, leaves. If there are any tiles left, but no further tile can be removed, we conclude that the polyomino cannot be constructed.

We begin by stating two facts about removable tiles. Firstly, by removing a removable tile, other tiles do not lose their property of being removable. And secondly, if a tree-shaped polyomino is constructible, then after removing any removable tile, the resulting polyomino is also constructible.

\begin{lemma}\label{lem:tree:property}
	Let $P$ be a tree-shaped polyomino and $\mathcal{R}_P$ the set of removable tiles of $P$. For all $P'\subseteq P$ it holds that if $t\in \mathcal{R}_P \cap P'$, then $t\in \mathcal{R}_{P'}$.
\end{lemma}

\begin{proof}
	Let $t, t' \in \mathcal{R}_P$ and $P' := P\setminus t'$. Due to the definition of $\mathcal{R}_P$, $P'$ is connected. Furthermore, the path that removes $t$ from $P$ still exists, because the removal of $t'$ cannot block this path. Therefore, $t \in \mathcal{R}_{P'}$.
\end{proof}

\begin{lemma}\label{lem:constructible-tree}
	Let $P$ be a constructible tree-shaped polyomino and let $t$ be a removable tile.
	Then $P\setminus t$ is also constructible.
\end{lemma}

\begin{proof}
	For the sake of a contradiction, let $P'\subsetneq P$ be the largest polyomino that results by removing tiles from $P$ such that $t\in P'$ and $P'\setminus t$ is constructible. 
	Let $t'$ be the last tile that was removed in this process to obtain $P'$ from $P$. 
	Because of \cref{lem:tree:property}, a feasible deconstruction sequence can contain two subsequent deconstruction steps so that the first one removes $t$ from the shape $P'\cup t'$, followed by removing $t'$. 
	Thus, $P'$ was not the largest constructible polyomino, which is a contradiction.
\end{proof}

By using \cref{lem:constructible-tree} iteratively, we obtain a simple strategy that decides whether a tree-shaped polyomino is constructible or not. By applying suitable subroutines and data structures, this yields a linear-time algorithm.

\begin{theorem}\label{thm:tree}
	Let $P$ be a tree-shaped polyomino of size $n$. It can be decided in $O(n)$ time whether or not $P$ is constructible.
\end{theorem}

\begin{proof}
	Notice that every free position within the workspace that has an $L_\infty$-distance greater than 2 to any tile of the shape has no effect on a potential indestructibility of that shape. 
	Thus, to keep track of the tiles that are and that become removable, we just have to consider every position within an $L_\infty$-distance of at most 2 to any tile of the shape. 
	These positions are marked as \texttt{blocked}. 
	Note that there are at most $O(n)$ \texttt{blocked} positions, because each tile has at most 24 positions in a distance of at most 2.
	
	\texttt{Blocked} positions that are not adjacent to $P$, but are reachable from the outside of its bounding box are marked as \texttt{unblocked}. This can be done by a simple graph scan algorithm like BFS or DFS in $O(n)$ time.
	
	To identify whether a tile is removable, we only need to check the following:
	(i)~$t$ is a leaf tile, and
	(ii)~at least one position of $N[t]$ is adjacent to an \texttt{unblocked} position. 
	This procedure costs $O(1)$ time per tile; thus $O(n)$ time in total.
	\begin{figure}
		\begin{subfigure}[b]{0.49\textwidth}\centering
			\includegraphics[page=14]{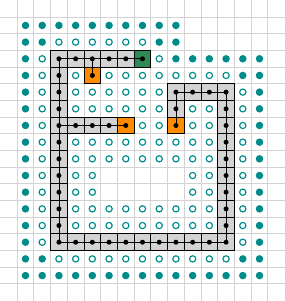}
			\caption{}
			\label{fig:leaf-removal-1}
		\end{subfigure}
		\begin{subfigure}[b]{0.49\textwidth}\centering
			\includegraphics[page=15]{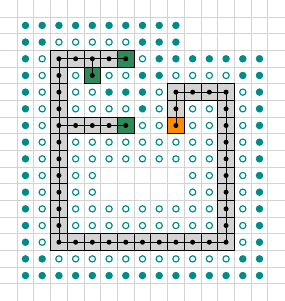}
			\caption{}
			\label{fig:leaf-removal-2}
		\end{subfigure}
		\caption{
			(a) A polyomino $P$ (gray) and the dual graph of $P$ (black). 
			Tiles in green are removable, and tiles in orange are leaves that are not removable yet. 
			Cyan dots represent free positions with the state \texttt{unblocked}, cyan circles represent free positions with the state \texttt{blocked}. 
			(b) The situation after the removal of the green tile in (a).}
		\label{fig:leaf-removal}
	\end{figure}
	
	The following steps are executed iteratively, until either all tiles have been removed, or no further tile can be removed from $P$. For this, let $t\in P$ be a removable tile.
	\begin{description}
		\item[Step 1:] Remove $t$ from $P$ and mark its former position as \texttt{blocked}.
		\item[Step 2:] \label{algo:tree:2} For each free position $p\in N[t]$ that is adjacent to an \texttt{unblocked} position:
		\begin{enumerate}[(a)]
			\item \label{algo:tree:2:a} Mark $p$ as \texttt{unblocked} if $N[p]$ is free and start a graph scan on \texttt{blocked} positions as well as mark them as \texttt{unblocked}, if they are reachable from $p$ and their respective neighborhood is free.
			\item \label{algo:tree:2:b} During the scan, if a \texttt{blocked} position does not change its state, check whether adjacent tiles become removable.
		\end{enumerate}
	\end{description}
	If there are tiles left that cannot be removed, we conclude that $P$ is not constructible. Otherwise, $P$ is constructible and the output is a feasible construction sequence.
	
	Each free position can change its state only once. 
	We charge the constant cost of the checking part in  
	Step 2(b) of the algorithm to the respective adjacent \texttt{unblocked} position. 
	Thus, over all iterations, each position is charged $O(1)$ cost, yielding a runtime of $O(n)$.
\end{proof}

It is easy to see that the same holds true for tree-shapes in 3D.

\begin{corollary}\label{cor:tree}
	Let $P$ be a tree-shaped polycube of size $n$. It can be decided in $O(n)$ time whether or not $P$ is constructible.
\end{corollary}

\subsection{Scaled shapes}
In the previous section we showed that it can be decided efficiently whether or not a tree-shape is constructible. 
A crucial point is that these shapes are thin. 
A simple example for a non-constructible tree-shape is a slightly modified \emph{Hilbert curve} (in fact it is even a path) in that both endpoints are not removable, see \cref{fig:hilbert-curve}. 
Therefore, we want to assume that we are allowed to construct a \emph{scaled copy} of the actual shape.

\begin{figure}
	\begin{subfigure}[b]{0.47\textwidth}\centering
		\includegraphics[page=18]{figures/figs-final}
		\caption{}
		\label{fig:hilbert2d}
	\end{subfigure}
	\begin{subfigure}[b]{0.52\textwidth}\centering
		\includegraphics[scale=0.045]{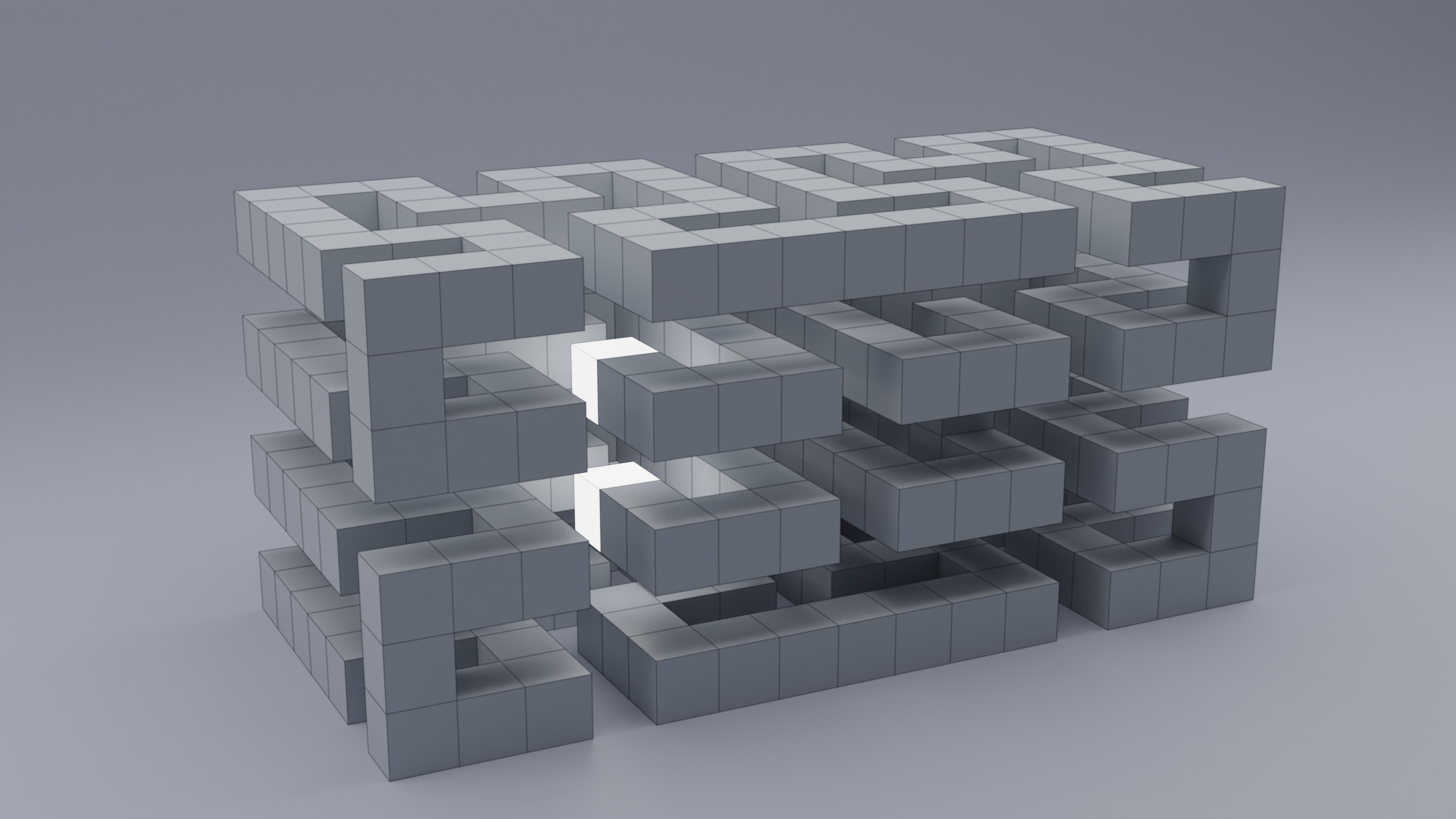}
		\caption{}
		\label{fig:hilbert3d}
	\end{subfigure}
	\caption{(a) A (slightly modified) 2D Hilbert curve that cannot be constructed because both endpoints are not removable. (b) A (slightly modified) 3D Hilbert curve that cannot be constructed; the two leaves are represented as luminous cubes.}
	\label{fig:hilbert-curve}
\end{figure}

\begin{definition}
	Let $P$ be a polyomino and $c\in \mathbb{N}$. By $P^c$ we denote the $c$-scaled copy of $P$, i.e., each tile in $P$ is replaced by a $c\times c$ square of tiles.
\end{definition}

We show that the 2-scaled copy of a non-degenerate polyomino is constructible.
Recall that in a non-degenerate polyomino $P$, for every non-adjacent pair $t_1, t_2\in P$ of tiles with $\vert N[t_1]\cap N[t_2]\vert \neq \emptyset$, there is at least one occupied position $p\in N[t_1]\cap N[t_2]$, i.e., the tiles $t_1, t_2$ have a common neighbor.
Note that no scaling factor suffices to guarantee constructibility for degenerate polyominoes, see \cref{fig:scaling-degenerate-shape}. 

\begin{figure}[ht]
	\centering
	\includegraphics[page=13, scale=0.8]{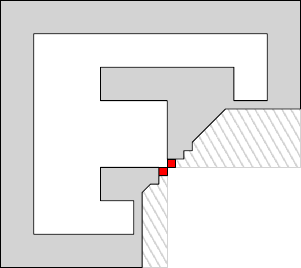}
	\caption{For deconstruction, it is necessary to remove at least one of the red tiles. Independent from the scaling factor, both tiles block each other and we are only able to deconstruct a staircase to the right and below both tiles (hatched part).}
	\label{fig:scaling-degenerate-shape}
\end{figure}

\begin{definition}\label{def:empty}
	We call a polyomino $P$ $c$-\emph{empty}, if for every pair $(p_i, p_j)$ of free positions in the same connected component of $G_{\mathbb{Z}^2\setminus P}$ there is a square of size $c\times c$ that initially overlaps with $p_i$ can be moved in such a way that it eventually overlaps with $p_j$ without overlapping with any position of $P$ at any time. 
	We call $P$ \emph{weakly $c$-empty} if at any time at most one corner of the square overlaps with a position of $P$.
\end{definition}

For an illustration of \cref{def:empty}, see \cref{fig:decomp_empty}\subref{fig:empty}. 
Note that in the case of weakly 3-emptiness it is sufficient to ensure that during the motion of a tile its neighborhood $N[\cdot]$ is kept empty. However, it is not possible to limit the definition of weakly 3-emptiness to this, because then free positions that are adjacent to two tiles (especially ``corner positions'' in the free space) would be excluded from this definition.
For a $3$-empty polyomino $P$, it is straightforward to see that as long as a tile $t\in P$ can be moved to a free position $p$ that lies in the outer face, such that all surrounding positions of $p$ are free as well, $t$ can be removed from $P$. 
The intuition is the following: Consider a path of a $3\times 3$-square, centered at that free position $p$, that connects $p$ to the outside of the bounding box of $P$. The deconstruction step for the tile $t$ consists of all positions given by that path.
Note that this still holds if we consider weakly 3-empty polyominoes.

We show that such a removable tile can always be found if $P$ is the 2-scaled copy of a non-degenerate polyomino.
One of the core ideas of our method is to make $P$ weakly 3-empty.
If $P$ is weakly 3-empty, we consider a partition of $P$ into horizontal \emph{slabs}.
Based on this partition, we show that either a leaf of the dual graph of the partition can be removed, or a hole of $P$ can be \emph{cut open}, i.e., two adjacent tiles of $P$ can be removed to reduce the number of holes by 1.

\begin{definition}
	A \emph{slice} of a shape $P$ is the set of all tiles sharing the same $y$-coordinate.
	A~\emph{slab} is a maximal connected set of tiles within a slice.
\end{definition}

\begin{definition}
	Let $P$ be a shape, and $\mathcal S_P$ its partition into slabs.
	For 2D shapes, or 3D shapes, we refer by $\mathcal C(\mathcal S_P)$ to the edge-contact graph, or to the face-contact graph of $\mathcal S_P$, respectively. That is, each slab is represented by a vertex, and two vertices are connected if and only if the union of both slabs is connected.
\end{definition}

\begin{figure}
	\centering
	\begin{subfigure}[b]{0.45\textwidth}
		\centering
		\includegraphics[page=4,scale=0.7]{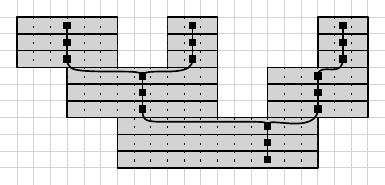}
		\caption{}
		\label{fig:horizontal_decomp}
	\end{subfigure}\hfill
	\begin{subfigure}[b]{0.55\textwidth}
		\centering
		\includegraphics[page=5,scale=0.7]{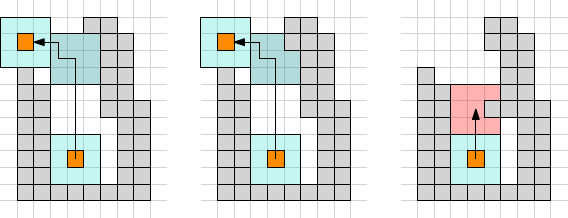}
		\caption{}
		\label{fig:empty}
	\end{subfigure}
	\caption{(a) The $3$-scaled copy $P^3$ of a polyomino $P$ (gray tiles), its partition into horizontal slabs $\mathcal{S}_P$, and its dual graph $\mathcal{C}(\mathcal{S}_P)$ (dark gray). (b) A polyomino that is 3-empty, weakly 3-empty, and not (weakly) 3-empty, respectively.}
	\label{fig:decomp_empty}
\end{figure}

It is straightforward to see that any leaf of $\mathcal{C}(\mathcal{S}_P)$ can be removed if it lies in the outer face and the polyomino is weakly 3-empty.

\begin{lemma}\label{lem:scaling:weakly}
	Let $P$ be a non-degenerate, weakly 3-empty polyomino. If $\mathcal{C}(\mathcal{S}_P)$ contains a~leaf that lies in the outer face, then the corresponding slab~$S$ can be removed from~$P$.
\end{lemma}

\begin{proof}
	Without loss of generality, we assume that there is no neighbor below $S$ and that $P\setminus t$ is connected, where $t$ is the rightmost tile of $S$. 
	Analogous arguments hold in case that $t$ is the leftmost tile of a slab $S$.
	
	Consider position of $t$ and the position $p$ which lies one step below and one step to the right of $t$. 
	Because $S$ is weakly 3-empty, $p$ is free and it is the center of a $3\times 3$ square that contains $t$ in the top left corner.
	Because $P$ is non-degenerate and by assumption the slab $S$ has no neighbor to the bottom side, we can move $t$ to $p$ by moving it one step down and then one step right. 
	Because $P$ is weakly 3-empty, $t$ can be removed. 
	It is straightforward to see that $P\setminus t$ is still weakly 3-empty.
	
	Because we can iteratively remove the leftmost or the rightmost tile of $S$ without losing connectivity, we eventually remove $S$ from $P$.
\end{proof}

\begin{theorem}\label{thm:scaling:simple:2D}
	For every non-degenerate polyomino $P$, its 2-scaled copy $P^2$ is constructible.
\end{theorem}

\begin{proof}
	The approach works in two phases. 
	The first phase removes tiles from $P^2$ such that the remaining shape is weakly 3-empty. 
	The second phase deconstructs the remaining shape slab by slab. 
	Any holes that may be contained in the shape are cut open at appropriate times within the deconstruction sequence at suitable positions. 
	If necessary, these phases are executed again one after the other.
	
	Phase 1 works as follows, see~\cref{fig:scale_fac2}.
	Consider the set $F$ that contains all free positions that can participate in a pair with a free position outside the bounding box of $P^2$, such that these pairs fulfill the weakly 3-empty property. 
	We will describe how tiles of $P^2$ can be removed such that any position lying in the outer face will belong to this set.
	
	As long as there are adjacent free positions that are not in $F$, repeat the following:
	Let $p_1,p_2\notin F$ be two adjacent free positions that both have a neighbor in $F$, say $p'_1$ and $p'_2$ respectively.
	Without loss of generality, let $(p_1,p_2)$ be a vertical pair and let the neighbors be in the left direction.
	Then, to the top and to the bottom of $(p'_1,p'_2)$ there are tiles $t_t,t_b\in P^2$ (or else $p_1, p_2\in F$).
	Consider the maximal horizontal line $L$ of $m$ tiles $(t_1,\dots,t_m)$ ordered from left to right with $t_b\in L$ that have a free position to the top.
	We make the following case distinction:
	\begin{description}
		\item[{Case} 1:] $t_b \neq t_1$: Remove the tile $t$ that is right of $t_b$ by moving $t$ two steps to the top. This position is the position to the left of $p'_2$, i.e., a position from $F$. Furthermore, this position has to be the center of a $3\times 3$ square containing only $t_t$.
		Thus, $t$ can be removed from $P^2$.
		After $t$ is removed, $t_b$ can be moved one step to the top, followed by a step right and a step up, reaching the exact same position.\label{item:2scaled2}
		\item[{Case} 2:] $t_b = t_1$: Remove $t_b$ by moving the tile one step to the top, followed by a step to the right.\label{item:2scaled1}
	\end{description}
	After removing $t_b$, we can proceed by iteratively removing $L$.
	Note that in some cases $t_1$ or~$t_m$ may not be removable, because they lie in a corner.
	However, the free positions adjacent to the top of these tiles will be contained in $F$, and thus, the corridor is wide enough.
	
	The cases when $p_1,p_2$ have their neighbors from $F$ in the direction to the right, or when they are a horizontal pair, are handled analogously. 
	To maintain connectivity, we always choose the line $L$ in a way that $L$ is to the left, or to the down direction of $(p'_1,p'_2)$.
	
	For Phase 2 we can assume that the remaining polyomino $P'$ is weakly 3-empty. Thus, we can apply Lemma~\ref{lem:scaling:weakly} to remove slabs that correspond to leaves.
	If there is no such slab, then we either removed every single tile, or there is a slab $S$ incident to the outer face with $k$ neighbors in $\mathcal C(\mathcal S_{P'})$ on one side only, whose removal would produce at most $k-1$ connected components.
	This slab must exist. 
	To see this, consider the set of all slabs  that have their neighbors in $\mathcal C(\mathcal S_{P'})$ on one side only, and consider a maximal path in $\mathcal C(\mathcal S_{P'})$ between these slabs. 
	Then the slab $S'$ is a slab with the desired property, or else we can extend the path.
	
	Now, consider a slab $S$ as defined above.
	We know that there are two slabs $S_1$ and $S_2$ adjacent to $S$ that belong to the same connected component in $P'\setminus S$.
	Without loss of generality, let $S_1$ be to the left of $S_2$.
	The first two tiles $t_1$ and $t_2$ from $S$ that have a larger $x$-coordinate than the rightmost tile in $S_1$ can be removed.
	There cannot be any tile above or below $t_1$ and $t_2$, because we consider the 2-scaled copy of a polyomino, and thus, this cannot break any connectivity.
	Also, any slab connected to $S$ to the left of $t_1$ remains connected to any slab that is connected to $S$ to the right of $t_2$.
	This implies that $P'\setminus\{t_1, t_2\}$ is a connected polyomino.
	
	However, $P'\setminus\{t_1, t_2\}$ may 
	no longer be weakly 3-empty.
	By restarting Phase 1, further tiles can be removed.
	Because tiles are getting removed in each case, $P^2$ is eventually deconstructed. Thus, we conclude that for every initial polyomino $P$, its $2$-scaled copy $P^2$ is constructible.
\end{proof}

\begin{figure}
	\begin{subfigure}[b]{0.495\textwidth}
		\centering
		\includegraphics[page=16]{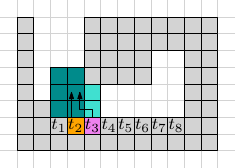}
		\caption{}
		\label{fig:scale_fac2-1}
	\end{subfigure}
	\begin{subfigure}[b]{0.495\textwidth}
		\centering
		\includegraphics[page=17]{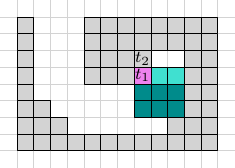}
		\caption{}
		\label{fig:scale_fac2-2}
	\end{subfigure}
	\caption{
		(a) This figure visualizes the first case in the proof of~\cref{thm:scaling:simple:2D}. The horizontal line $L=(t_1,\dots, t_8)$ with $t_b = t_3$. 
		Cyan and teal positions must be free, teal positions correspond to the pair $(p'_1,p'_2)$. 
		The orange tile is removed first, followed the purple tile. 
		(b) A visualization of the second case, i.e., when $t_b = t_1$. 
		We can immediately remove~$t_1$. Note that this figure is rotated by 90~degrees compared to the arguments given in the proof.}
	\label{fig:scale_fac2}
\end{figure}

Because there are (even tree-shaped) indestructible polyominoes, this result is tight. 
If the polyomino is already 3-empty, we can skip the first phase. 
Whenever we have to cut open a hole, we remove three (instead of two) tiles, such that the property of being 3-empty is preserved.
This results in the following corollary.

\begin{corollary}\label{lem:scaling:simple:3-empty}
	Every non-degenerate, 3-empty polyomino is constructible.
\end{corollary}

In the three-dimensional setting we cannot simply widen narrow corridors when scaled by a factor of 2, e.g., imagine a tube with an inner diameter of 2.
Therefore, the approach in 2D does not work in 3D. 
However, we show that a scaling factor of 3 suffices to guarantee constructibility.

\begin{definition}
	Let $S$ be a slab of a 3D shape, and $S_1, S_2$ be two slabs adjacent to~$S$.
	$S_1$ is called \emph{ring} if  there are tiles of $S_1$ bounding a (3D) hole of $P$.
	$S_1$ and $S_2$ build a \emph{loop} (or 2D hole), if both are in the same connected component of $P\setminus S$.
	If both cases do not apply, $S_1$ is called a \emph{pillar}.
\end{definition}

Using the same arguments as in the 2D case, we obtain the following lemma.

\begin{lemma}\label{lem:scaling:loop}
	Let $P$ be a 3D shape and $\mathcal{S}_P$ a decomposition of $P$ into slabs.
	If $\mathcal{S}_P$ contains no slab of degree one, then there exists a slab accessible from the outer face with at most one pillar (outside a 3D hole).
\end{lemma}

\begin{proof}
	Consider the set of all slabs adjacent to the outer face, that have their neighbors on one side only.
	In this set, there is a maximal path of these slabs using pillars that are adjacent to the outer face.
	Because $P$ is finite, the start and end slab of this path can only have one pillar.
\end{proof}

\begin{theorem}\label{thm:scaling:simple:3D}
	For every non-degenerate polycube $P$, its 3-scaled copy $P^3$ is constructible.
\end{theorem}

\begin{proof}
	Similar to the 2D case, we proceed in two phases.
	In Phase 1, we successively remove slabs that are leaves in $\mathcal C(\mathcal S_P)$.
	Because we consider the 3-scaled copy of a shape, there is always sufficient space to remove all tiles from these slabs (see also 2D algorithm). 
	If every slab has degree at least two, we proceed with Phase 2. Note that leaves can still exist, but these are not adjacent to the outer face and therefore not removable.
	
	In Phase 2, we first search for a slab $S$ that (i) is adjacent to the outer face, (ii) has neighbors  on one side only, and (iii) is adjacent to a ring $S_r$.
	Consider a hole $H$ that is bounded by tiles of $S$ and $S_r$.
	Because $S$ has neighbors on one side only, $S$ covers one face of $H$ completely.
	Let $S_r^\downarrow$ be the tiles from $S$ that are adjacent to $S_r$.
	Then a $3\times 3$ square from $S$ that is adjacent to $H$ and $S_r^\downarrow$ can be removed~(see also \cref{fig:3D_scaling_defs_a}).
	
	In the case that $S$ does not exist, then, because of \cref{lem:scaling:loop}, there must be at least one slab $S'$ with at least one adjacent loop.
	Let $S_1, S_2$ be two slabs adjacent to $S'$ building a loop, and let $S_1', S_2'$ be the set of tiles from $S'$ adjacent to $S_1$ and $S_2$, respectively.
	Consider the set $C$ of maximal connected components of $S\setminus \{S_1', S_2'\}$.
	The components of $C$ can be divided into the following three types (see also \cref{fig:3D_scaling_defs_b}):
	
	\begin{description}
		\item[{Type} 1:] Components that are adjacent only to tiles that belong to $S_1'$.
		\item[{Type} 2:] Components that are adjacent only to tiles that belong to $S_2'$.
		\item[{Type} 3:] Components that are adjacent only to tiles that belong to both $S_1'$ and $S_2'$.
	\end{description}
	
	\begin{figure}
		\centering
		\begin{subfigure}{.49\columnwidth}
			\includegraphics[scale=.5]{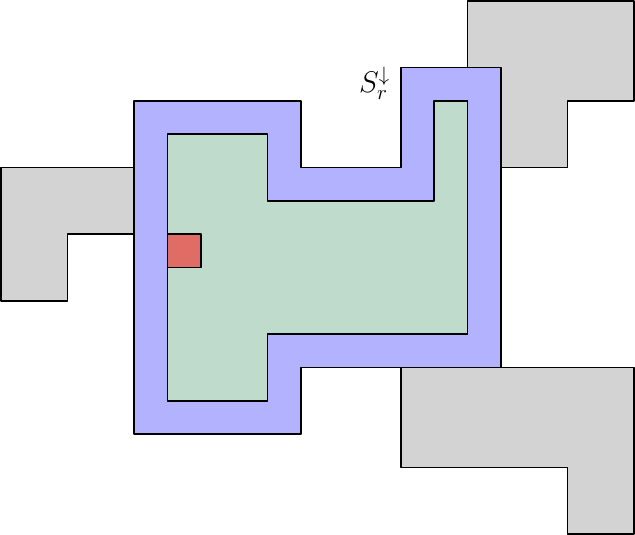}
			\caption{}
			\label{fig:3D_scaling_defs_a}
		\end{subfigure}\hfill
		\begin{subfigure}{.49\columnwidth}
			\includegraphics[scale=.5]{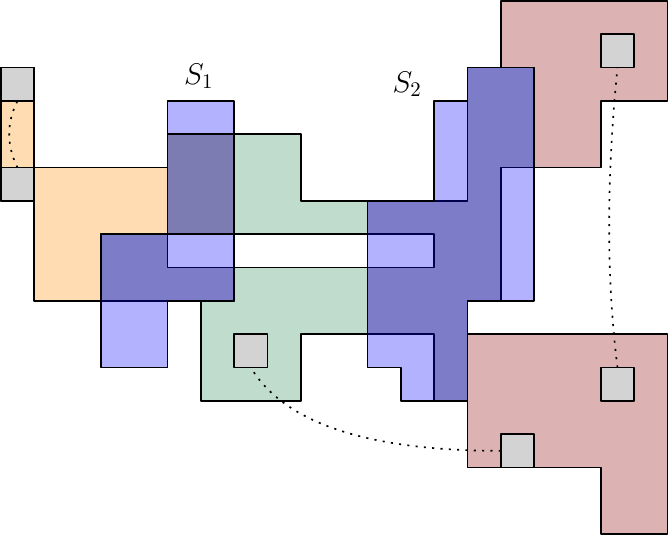}
			\caption{}
			\label{fig:3D_scaling_defs_b}
		\end{subfigure}
		\caption{
			(a) A slab $S$ with tiles $S_r^\downarrow$ adjacent to a ring (shown in blue) and adjacent to a hole~(green).
			We remove a $3\times 3$ square from the green area (shown in red).
			Gray area is part of $S$ but is not needed to consider here.
			(b) Blue area corresponds to slabs $S_1$ and $S_2$, dark blue area corresponds to tiles from $S$, that are below $S_1$ and $S_2$, i.e., $S_1'$ and $S_2'$.
			The orange area corresponds to a connected component of Type 1, red areas to components of Type 2, and the green areas to components of Type 3.
			Gray squares indicate loops, when connected by a dotted curve.}
		\label{fig:3D_scaling_defs}
	\end{figure}
	
	Any component that has no adjacent slab above/below itself can be removed.
	In the case that thereby all components of Type 3 have been removed, we proceed again with Phase 1.
	Otherwise, we make the following case distinction:
	
	\begin{description}
		\item[{Case} 1:] Only components of Type 1 and Type 3 are still present.
		In this case, we remove $S_2'$ and all tiles adjacent to $S_2'$. 
		This keeps the remaining shape connected and non-degenerate, while also giving sufficient space to remove $S_2$ if it becomes a degree-one slab.
		
		\item[{Case} 2:] Only components of Type~2 and Type~3 are still present.
		This is analogue to Case~1.
		
		\item[{Case} 3:] Only components of Type~3 are present.
		This is analogue to Case~1.
		
		\item[{Case} 4:] Components of all three types are still present.
		Without loss of generality, assume that the only pillar (if existing) is adjacent to a component of Type~1 or Type~3; otherwise we switch $S_1$ and $S_2$.
		Consider all pairs of slabs adjacent to $S'$ that build a loop.
		Suppose there is a loop build by slabs $\bar S_1$ and $\bar S_2$ both adjacent to components of Type~2 (see rightmost gray squares in the red area in \cref{fig:3D_scaling_defs_b}).
		Then, we restart Phase~2 with slab $S'$, and the loop build by $\bar S_1$ and $\bar S_2$, such that the former slabs $S_1$ and $S_2$ are then adjacent to components of Type~1 or Type~3.
		This ensures that we do not cycle back to these two slabs.
		
		We repeat this procedure until we have the following situation:
		(i) Every loop beginning in a component of Type~2 has its end in a component of Type~1 or Type~3, and (ii) the only pillar (if existing) is connected to a component of Type~1 or Type~3.
		With that, we can simply remove any component of Type~2 without losing connectivity, and we can continue with Case~1.
	\end{description}
	In each iteration of this phase, at least one loop is getting removed or one hole is cut open.
	Therefore, Phase 2 can only have finitely many iterations and thus eventually terminates.
	Combining both phases shows that the 3-scaled copy of every non-degenerate 3D shape is constructible.
\end{proof}

Unlike as in the 2D case, we cannot show that this result is tight and conjecture that even the 2-scaled copy of every non-degenerate polycube is constructible.

\section{Conclusion and future work}
We provided a number of algorithmic results for assembling shapes by connecting particles one after the other to a fixed seed tile in a single step model. 
For future research several interesting problems remain open.
We showed that constructibility can be decided for special classes of shapes, but we do not know how hard it is to decide whether an arbitrary polyomino is constructible or not in the considered model, i.e., what is the computational complexity of 2D-\problemtitle? It seems that a similar construction as in the three-dimensional setting should work; however, the main difficulty is guaranteeing connectivity during a feasible deconstruction. We note that this is also an open question in the full tilt model~\cite{becker2018tilt}. Another question is whether 3D-\problemtitle remains NP-complete when restricted to the class of non-degenerate shapes.

In this paper we add one tile after the other to an assembly. 
If this assumption is relaxed, i.e., more than one tile at a time can be added to and controlled in the workspace, it is easy to see that more shapes are constructible, e.g., see \cref{fig:more-tiles-construction}.
Is there a classification of shapes that can be built in this model?
This also leads to the question: ``Which shapes are constructible by using pre-assembled shapes (e.g., trominoes, tetrominoes, etc.)?''
By taking this a step further, we could also ask for a staged approach similar to~\cite{schmidt2018efficient}, where whole subassemblies can attach to each other. Another question arises by considering multiple seed tiles. What classes of shapes are constructible, if multiple seed tiles can be placed in advance, again see~\cref{fig:more-tiles-construction}.

\begin{figure}[ht]
	\centering
	\begin{subfigure}[b]{0.24\textwidth}
		\includegraphics[page=10]{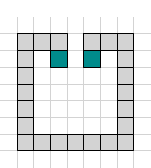}
		\caption{}
		\label{fig:more-tiles-construction1}
	\end{subfigure}\hfill
	\begin{subfigure}[b]{0.24\textwidth}
		\includegraphics[page=7]{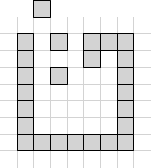}
		\caption{}
		\label{fig:more-tiles-construction2}
	\end{subfigure}\hfill
	\begin{subfigure}[b]{0.24\textwidth}
		\includegraphics[page=8]{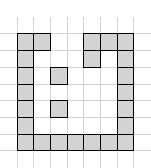}
		\caption{}
		\label{fig:more-tiles-construction3}
	\end{subfigure}\hfill
	\begin{subfigure}[b]{0.24\textwidth}
		\includegraphics[page=9]{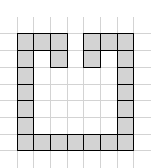}
		\caption{}
		\label{fig:more-tiles-construction4}
	\end{subfigure}
	\caption{It is easy to see that the shape shown in (a) is not constructible if only one tile at a time is allowed to be controlled. Because tiles stick together when they are on adjacent positions, no leaf (dark cyan tiles) can be removed, as well as no other tile without losing connectivity. If~we instead are allowed to add multiple tiles simultaneously, the polyomino can be constructed as seen in \cref{fig:more-tiles-construction2,fig:more-tiles-construction3,fig:more-tiles-construction4} by two down tilts followed by three up tilts. The polyomino is also constructible if we are allowed to place two (dark cyan colored) seed tiles in advance. Note that in this case even adding one tile at a time is sufficient to construct the polyomino.}
	\label{fig:more-tiles-construction}
\end{figure}

Is even the 2-scaled copy of every non-degenerate polycube constructible? Is there a similar example as in \cref{fig:scaling-degenerate-shape} for degenerate polycubes, i.e., is there a polycube such that no scaling factor guarantees constructibility in the case of degeneracy?

Our considerations are based on a tile with a single glue type on all sides so that tiles immediately stick together when they are on adjacent positions. It is a natural following question to ask what changes once we have different glue types.

\newpage
\bibliography{main}

\end{document}